\documentclass{LMCS}

\def\dOi{10(3:4)2014}
\lmcsheading%
{\dOi}
{1--44}
{}
{}
{Jan.~\phantom06, 2012}
{Aug.~15, 2014}
{}

\ACMCCS{[{\bf Theory of computation}]: Formal languages and automata
  theory---Automata over infinite objects}

\usepackage{amssymb,mathrsfs}
\usepackage{amsmath,amscd}
\usepackage[all]{xy}
\usepackage{mik-makro,pdfsync,hyperref}

\begin{document}
\newcommand{\paragrafik}[1]{\vspace{0.2cm} \noindent {\bf #1.}}

\newcommand{\biseq}{\sim}
\newcommand{\Franosp}{Fra{\"i}ss{\'e}}
\newcommand{\Fra}{\Franosp{} }
\newcommand{\fra}{\Franosp{} }

\newcommand{\D}{\mathbb{D}}
\renewcommand{\P}{\mathcal{P}}
\newcommand{\N}{\mathbb{N}}
\newcommand{\Q}{\mathbb{Q}}
\newcommand{\Z}{\mathbb{Z}}
\newcommand{\permgrp}{\mathop{\mathrm{Sym}}}
\newcommand{\Set}{\textbf{Set}}
\newcommand{\GSet}{G\textbf{-Set}}
\newcommand{\GNom}{G\textbf{-Nom}}
\newcommand{\GSetI}{G\textbf{-Set\textsuperscript{1}}}
\newcommand{\GNomI}{G\textbf{-Nom\textsuperscript{1}}}

\renewcommand{\choose}[2]{
  \bigl(\begin{smallmatrix}
  #1 \\ #2
  \end{smallmatrix} \bigr)}
  
\newcommand{\sem}[1]{[\![#1]\!]}
\newcommand{\cossem}[1]{\sem{#1}^{\mathsf{c}}}
\newcommand{\supgsem}[1]{\sem{#1}^{\mathsf{e}}}
\newcommand{\supssem}[1]{\sem{#1}^{\mathsf{ec}}}
\newcommand{\fsem}[1]{\sem{#1}}

\newcommand{\absclass}[2]{[#2]_{#1}}

\newcommand{\struct}{\mathfrak A}
\newcommand{\structA}{\mathfrak A}
\newcommand{\structB}{\mathfrak B}
\newcommand{\structC}{\mathfrak C}
\newcommand{\structD}{\mathfrak D}
\newcommand{\structM}{\mathfrak{M}}
\newcommand{\structN}{\mathfrak{N}}
\newcommand{\KK}{\mathcal{K}}
\newcommand{\U}{{\mathfrak U}_{\KK}}
\newcommand{\GU}{G_{\KK}}
\newcommand{\DU}{{\mathbb D}_\KK}
\newcommand{\autgrp}{\mathop{\mathrm{Aut}}}
\newcommand{\isoeq}{\equiv_{S,T}}
\newcommand{\KRepr}{\KK\textbf{-Repr}}

\newcommand{\emptyword}{\varepsilon}
\newcommand{\GDFA}{$G$-DFA}
\newcommand{\GUDFA}{$\GU$-DFA}
\newcommand{\GNFA}{$G$-NFA}

\newcommand{\removed}[1]{}

%% due to the dependence on amsart.cls, \begin{document} has to occurxt
%% BEFORE the title and author information:

%\title[short title]{Fra\"iss\'e Schma\"iss\'e}
\title[Automata theory in nominal sets]{Automata theory in nominal sets\rsuper*}

\author[M.~Boja\'nczyk]{Miko\l{}aj Boja\'nczyk}	%required
\address{University of Warsaw}	%required
\email{\{bojan, klin, sl\}@mimuw.edu.pl}  %optional
%\thanks{thanks 1, optional.}	%optional

\author[B.~Klin]{Bartek Klin}	%optional
\address{\vspace{-18 pt}}	%optional
%\email{klin@mimuw.edu.pl}  %optional
%\thanks{thanks 2, optional.}	%optional

\author[S.~Lasota]{S{\l}awomir Lasota}	%optional
\address{\vspace{-18 pt}}	%optional
%\email{sl@mimuw.edu.pl}  %optional
%\thanks{thanks 3, optional.}	%optional

%% etc.

%% required for running head on odd and even pages, use suitable
%% abbreviations in case of long titles and many authors:

%% mandatory lists of keywords and classifications:
\titlecomment{{\lsuper*}An extended abstract of this paper appeared in the proceedings of LICS'11~\cite{BKL11}.}
\titlecomment{This work is supported by the Polish National Science Centre (NCN) grant 2012/07/B/ST6/01497 (the last two authors)
and by the ERC Starting Grant \emph{Sosna} (the first author).}
%%%%%%%%%%%%%%%%%%%%%%%%%%%%%%%%%%%%%%%%%%%%%%%%%%%%%%%%%%%%%%%%%%%%%%%%%%%

%% the abstract has to PRECEED the command \maketitle:
%% be sure not to issue the \maketitle command twice!

\begin{abstract}
  \noindent
  We study languages over infinite alphabets equipped with some structure that can be tested by recognizing automata. We develop a framework for studying such alphabets and the ensuing automata theory, where the key role is played by an automorphism group of the alphabet. In the process, we generalize nominal sets due to Gabbay and Pitts.
  \end{abstract}

\maketitle

\tableofcontents

\section{Introduction}
We study languages and automata over infinite alphabets. Each alphabet comes with some structure that can be accessed by recognizing devices such as automata. Examples of such structures include:
\begin{itemize}
	\item {\it Equality.} There is an infinite set $\D$ whose elements are called \emph{data values}. Words are elements of $\D^*$, or in some cases $(\Sigma \times \D)^*$, for some finite set $\Sigma$. There is no structure on the data values except for equality. A typical language is 
		\begin{align*}
			\set{d_1 \cdots d_n \in \D^* : d_{i+1} \neq  d_{i} \mbox{ for all $i \in \set{1,\ldots,n-1}$}}.
		\end{align*}
		
	\item {\it Total order.} The set of data values is equipped with a total order.  A typical language is
		\begin{align*}
			\set{d_1 \cdots d_n \in \D^* : d_{i+1} > d_{i} \mbox{ for all $i \in \set{1,\ldots,n-1}$}}.
		\end{align*}		
	\end{itemize}
One could also consider data values equipped with a graph structure (where, e.g., the language of finite paths can be considered), a partial order etc. 

Note that the above descriptions do not determine the data values uniquely. 
One of the themes in this paper is the use of ``universal'' alphabets to obtain well-behaved notions of automata.

A device can only access data values through the given structure (e.g.~the equality or order relation).  For instance, in the case of data values with equality,  an automaton that accepts a two-letter word $de$ with $d\neq e$, will also necessarily accept the word $de'$ for any $e' \neq d$. 

The notion of structure on an alphabet is naturally captured by the group of its automorphisms. For example, in the case of unordered data values, the group consists of all bijections on $\D$. In the case of totally ordered data values, it is the group of all monotone bijections on $\D$.

In general, we work with a set of data values $\D$, together with a group $G$ of bijections of $\D$, which need not be the group of all bijections of $\D$. Such a pair $(\D,G)$ is called a \emph{data symmetry}. We then study sets $X$ which are acted upon by the group $G$.  A key example is the set $X=\D^*$, where $G$ acts separately on each letter. As far as languages are concerned, we work with languages $L \subseteq \D^*$ that are closed under actions of the group $G$.

\subsection{Contribution}
We now outline the main contributions of this paper.

\paragrafik{Nominal sets for arbitrary symmetries}
When working with a data symmetry $(\D,G)$ and a set $X$ with an action of $G$, we pay attention to the interplay between the canonical action of $G$ on $\D$ and the action of $G$ on $X$. An example of this interplay is the definition of a nominal set. A set $X$ is called  nominal wrt.~the symmetry if for every $x \in X$ there exists a \emph{finite} set of data values $C \subseteq \D$,  called a \emph{support of $x$}, such  that every $\pi \in G$ satisfies
\begin{align*}
	\forall{c \in C}.\ \pi(c)=c  \qquad \Rightarrow \qquad x\cdot\pi=x.
\end{align*}
The left side of this implication uses the canonical action of $\pi$ on $\D$, and the right side uses an action of $\pi$ on $X$. The intuition is that $x$ depends only on data values from $C$.

An example of a nominal set is $\D^*$, regardless of $G$: a support of a word  can be chosen as the set of letters that appear in the word. In the case of data values with equality, where $\D$ is a countably infinite set and $G$ is the group of all bijections on $\D$, the theory of nominal sets was developed by Gabbay and Pitts~\cite{GP02,pitts-book}. One of the contributions of this paper is a concept of nominal sets in different symmetries. %In particular, we give a sufficient condition for a symmetry to have least supports. This condition is satisfied by several interesting symmetries, in particular for those which model totally ordered data values and partially ordered data values.

\paragrafik{Automata theory in arbitrary symmetries} We study the theory of automata in various symmetries. For basic definitions of automata and languages, we transfer classical definitions to the world of nominal sets.
A crucial aspect here is an appropriate choice of the notion of 'finiteness'. As far as nominal sets are
considered, the appropriate notion is \emph{orbit finiteness}. Thus the abstract definitions we work with are just the classical definitions,
in which the requirement of finiteness (of alphabet, state space, etc.) is relaxed to orbit finiteness. 

%an abstract categorical theory of automata (see e.g.~\cite{Adamek:1990:AAC:575450}) to categories of nominal sets.
%We thus connect notions of automata developed within category theory with some models that have been developed by a%utomata theorists. 

It turns out that, in the cases of unordered and ordered data values,  the abstract definitions are expressively equivalent with existing definitions of finite memory automata~\cite{FK94,DL09} and register automata over totally ordered data~\cite{BLP10,FHL10}. 
Some minor adjustments to  finite memory automata  are needed; in fact, they help to make the automaton model robust. For instance, independently of the data symmetry, our models admit minimization of deterministic automata.
As one of our contributions, we provide an infinite-alphabet counterpart of the Myhill-Nerode theorem, thus concluding previous work on this theme~\cite{FK03,BLP10}.

\paragrafik{Effective representation} Our framework can be applied far beyond the theory of deterministic automata. 
We introduce a method of representing orbit finite nominal sets, together with relations and functions on them. We prove that an effective representation is possible in any symmetry of a certain form. 
As a result we obtain a toolkit which may be used to define and study nominal
nondeterministic or alternating automata, context-free grammars, pushdown automata, Petri nets, Turing machines or
many other natural models of computation.

%\subsection{Computability}
%An important ingredient of automata theory is algorithmics.
%We formally do not address this issue in this paper.
%However, we would like to emphesize that our finite representations of orbit finite nominal structures
%can be manipulated by algorithms in essentially the same way as classical finite structures.
%As an example, an algorithm may decide emptiness of a given nondeterministic automaton, 
%or test if a deterministic automaton is minimal, assuming that the state space of an automaton is orbit finite,
%but possibly infinite. 
%A crucial property of our representation is that the algorithms may be uniform, that is they work, roughly,
%in an arbitrary symmetry. 
%This immediately raise the question: what is the notion of computability in nominal sets?
%Pursuing this program is out of scope of this paper; we left it as a natural continuation of our work.
%Some initial steps in this line of research have been made in~\cite{BBKL12}.

\subsection{Background}
We briefly overview some related work on nominal sets in the context of automata theory.

\paragrafik{Nominal sets}
The theory of nominal sets originates from the work of Fraenkel in 1922, further developed by Mostowski in the 1930s. At that time, nominal sets were used to prove independence of the axiom of choice and other axioms. In Computer Science, they have been rediscovered by Gabbay and Pitts in~\cite{GP02}, as an elegant formalism for modeling name binding. Since then, nominal sets have become a lively topic in semantics; see~\cite{pitts-book} for a recent comprehensive study. They were also independently rediscovered by the concurrency community, as a basis for syntax-free models of name-passing process calculi, see~\cite{P99,MP05}.

\paragrafik{Automata for infinite alphabets}
Languages over infinite alphabets are a lively topic in the automata community. Two principal sources of motivation are XML and verification.  An XML document is often  modeled as a  tree with labels from the (infinite) set of all Unicode strings that can appear as attribute values. In software verification, the infinite alphabet can refer to pointers or function parameters.

Many automata models have been developed for infinite alphabets, including: finite memory automata~\cite{FK94}, automata for ordered data values~\cite{BLP10}, two-way automata and automata with pebbles~\cite{DBLP:conf/mfcs/NevenSV01}, alternating register automata~\cite{DL09}, data automata~\cite{DBLP:conf/lics/BojanczykMSSD06}, etc. See~\cite{DBLP:conf/csl/Segoufin06} for a survey.
%; the above list does  not even touch on timed automata.
%As demonstrated by the large number of existing models, 
There is no consensus as to which one is the ``real'' analogue of regular languages in the case of infinite alphabets. This question is a topic of debate, see e.g.~\cite{DBLP:conf/mfcs/NevenSV01} or~\cite{DBLP:journals/tcs/BjorklundS10}.

\paragrafik{Nominal sets and HD-automata} 
Nominal sets, studied until now in the case of unordered data values, %and known in this case
%as the Fraenkel-Mostowski permutation model of set theory
are a convenient tool for capturing name generation and binding.
They were introduced by Gabbay and Pitts~\cite{GP02} as a  mathematical model of name-binding and $\alpha$-conversion. 

A fruitful line of research starting from~\cite{P99} (see also~\cite{MP05} for an overview) uses a category equivalent to nominal sets for defining history-dependent (HD) automata, a syntax-free model of process calculi that create and pass names, like $\pi$-calculus. These are closely related to the notions of automata studied here. In fact, our representation of nominal sets, and consequently our notions of automata, are inspired by, and generalize, similar results for Gabbay-Pitts nominal sets as developed in~\cite{gadducci-etal,staton-thesis}. An initial connection between HD-automata and finite memory automata was made in~\cite{ciancia-tuosto}.

%In~\cite{AGMOS04} moving to category of nominal sets allowed the authors to obtain a fully-abstract game semantics for $\nu$-calculus, an extension of simply typed lambda calculus with names.

\paragrafik{Data monoids} One of us used group actions in formal language theory for infinite alphabets in~\cite{datamonoids}, which is the closest relation to our current work. That paper already includes: a group action of bijections of data values on languages, a central role of  finite supports, Myhill-Nerode congruence in the monoid setting. However, the main focus of~\cite{datamonoids} is the development of a monoid theory, including Green's relations and an effective characterization of first-order definable word languages. The present paper has a more fundamental approach. In particular we study: the connection with the literature on nominal sets,  different kinds of alphabets,  algorithms and methods of representing sets.

\subsection{Structure of the paper}

The remainder of this paper is divided in two parts. The first part, comprising
Sections~\ref{sec:gsets}--\ref{sec:others}, is about nominal sets in an arbitrary 
data symmetry and the basics of automata theory developed in orbit finite nominal sets, in place of 
finite classical sets. 
In the last two sections we briefly venture beyond orbit regular languages: we define context-free 
nominal languages and pushdown automata, prove them equivalent, and discuss possible 
further work and other models of computation that can be expressed in nominal sets.
The second part of the paper, spanning Sections~\ref{sec:gset-repr} to~\ref{sec:Fraisseautomata}, introduces finite representations
for orbit finite nominal sets, with an application to deterministic automata. 

One can also view this paper as an interleaving of two main threads. The first one comprises 
Sections~\ref{sec:gsets}, \ref{sec:nominals} and \ref{sec:gset-repr}-\ref{sec:fraisse}. 
In this thread, we study nominal sets for arbitrary symmetries and
prove finite representation theorems for orbit finite nominal sets, without reference to automata theory except as a source of examples. Under progressively stronger assumptions on the symmetries involved, we are able to obtain more concrete representations, culminating in the notion of a well-behaved \Fra symmetry in Section~\ref{sec:fraisse}.

The second thread is the development of rudiments of automata theory in nominal sets, which is done is 
Sections~\ref{sec:gautomata}, \ref{sec:nominalgautomata}-\ref{sec:others} and~\ref{sec:Fraisseautomata}. There,  we
define the notion of nondeterministic finite automaton in nominal sets,
prove the Myhill-Nerode theorem for deterministic automata,
relate our notion to finite memory automata of Kaminski and Francez~\cite{FK94,DL09},
and finally apply finite representation theorem for orbit finite automata in Section~\ref{sec:Fraisseautomata}. 

This paper is an extended and revised version of~\cite{BKL11}. We are grateful to Thomas Colcombet for suggesting that we use \Fra limits, and to Tomasz Wysocki for noticing Lemma 10.8(3).

\part{Nominal sets and automata}

\section{Group actions and data symmetries} \label{sec:gsets}

\paragrafik{Group actions}
A (right) action of a group $G$ on a set $X$ is a function $\cdot\,: X\times G\to X$,
written infix, subject to axioms
\[
	x\cdot e = x \qquad\qquad x\cdot(\pi\sigma) = (x\cdot\pi)\cdot\sigma
\]
for $x\in X$ and $\pi,\sigma\in G$, where $e$ is the neutral element of $G$. A set equipped with such an action is called a {\em $G$-set}.

\begin{exa}
Any set $X$ is a $G$-set with a trivial action defined by $x\cdot\pi=x$. The set $G$ can be seen as a $G$-set either with the composition action ($\pi\cdot\sigma=\pi\sigma$) or with the conjugacy action ($\pi\cdot\sigma=\sigma^{-1}\pi\sigma$). For any $G$-sets $X,Y$, the Cartesian product $X\times Y$ and the disjoint union $X+Y$ are $G$-sets with actions defined point-wise and by cases, respectively. 
\end{exa}

For further examples, we introduce the following:

\begin{defi}\label{def:symmetry}
A {\em data symmetry} $(\D,G)$ is a set $\D$ of {\em data}, together with a subgroup $G\leq\permgrp(\D)$ of the symmetric group on $\D$, i.e., the group of all bijections of $\D$.
\end{defi}

\begin{exa}\label{ex:symmetries}
We give names to a few important symmetries:
\begin{itemize}
\item the {\em classical symmetry}, where $\D=\emptyset$ and $G$ is the trivial group,
\item the {\em equality symmetry}, where $\D$ is a countably infinite set, say the natural numbers, and $G=\permgrp(\D)$ is the group of all bijections of $\D$,
\item the {\em total order symmetry}, where $\D=\Q$ is the set of rational numbers, and $G$ is the group of monotone bijections,
%\footnote{In Section~\ref{sec:fraisse} it will become apparent why we chose the rational numbers and not some other totally ordered set.},
\item the {\em integer symmetry}, where $\D=\Z$ is the set of integers, and $G$ is the group of translations $i\mapsto i+c$, isomorphic to the additive group of integers. We shall use this symmetry as a source of pathological counterexamples.
\end{itemize}
\end{exa}

\begin{exa}\label{ex:gsets}
For any data symmetry $(\D,G)$, a simple example of a $G$-set is the set  $\D$ itself, with the action defined by $d\cdot\pi = \pi(d)$. The action of $G$ on $\D$  extends pointwise to actions of $G$ on tuples $\D^n$, words $\D^*$, infinite words $\D^\omega$, or sets $\P(\D)$.

Other interesting $G$-sets include 
\begin{align*}
	\D^{(n)} &= \{(d_1,\ldots,d_n) : d_i\neq d_j \mbox{ for } i\neq j\}, \\
	\choose{\D}{n} &= \{C\subseteq \D : |C|=n\},
\end{align*}
with G-actions inherited from $\D$. For a subset $C \subseteq \D$, there is a $G$-set
\[
\D^{(C)} = \set { \pi|C \ : \ \pi \in G }.
\]
In other words, this is the set of all injective functions from $C$ to $\D$ that extend to some permutation from $G$.
The action is by composition:
\[
(\pi | C) \cdot \rho = (\pi \rho) | C.
\]

For the total order symmetry, one may also consider e.g.
\[
	\D^{(<n)} = \{(d_1,\ldots,d_n) : d_i< d_{i+1} \mbox{ for } 1\leq i<n \},
\]
and for the integer symmetry,
\[
	\Z_n = \{0,1,\ldots,n-1\}
\]
with action $k\cdot m = (k+m)\mod n$.
\end{exa}

\paragrafik{Orbits}
For any $x$ in a $G$-set $X$, the set 
\begin{align*}
	x\cdot G = \{x\cdot\pi\mid \pi\in G\}\subseteq X
\end{align*}
 is called the {\em orbit} of $x$. Any $G$-set is partitioned into orbits in a unique way. We will mostly be interested in {\em orbit finite} sets, i.e., those that have a finite number of orbits. In the world of $G$-sets these play the role of finite sets.
 
In group-theoretic literature, $G$-sets with only one orbit are called {\em transitive}. We prefer to call them simply {\em single-orbit} sets.

\begin{exa} \label{ex:orbits}
In the equality symmetry, elements of the powerset $\P(\D)$ are in the same orbit if and only if they have the same cardinality. As a result, $\P(\D)$ is not orbit finite.

In the equality symmetry, the set $\D^2$ has two orbits: 
\[
	\set{(d,d) : d\in\D} \qquad \set{(d,e) : d\neq e\in \D}
\]
In the total order symmetry, $\D^2$ has three orbits:
\[
	\set{(d,d) : d\in\D} \qquad \set{(d,e) : d<e\in \D} \qquad \set{(d,e) : e<d\in \D}
\]
In the integer symmetry, $\D^2$ is not orbit finite. Indeed, for any $y\in\D$, the set
\[
	\set{(x,x+y) : x\in\D}
\]
is a separate orbit.

In any symmetry, the set $\D^C$ has one orbit.
\end{exa}

\paragrafik{Equivariant relations and functions}
Suppose that $X$ is a $G$-set. A subset $Y \subseteq X$ is called \emph{equivariant} if it is preserved under group actions, i.e.~$Y \cdot \pi = Y$ holds for every $\pi \in G$. In other words, $Y$ is a union of orbits in $X$. This definition extends to the notion of an \emph{equivariant relation} $R \subseteq X \times Y$, by using the action of $G$ on the Cartesian product, or to relations of greater arity, by using the point-wise action of $G$.  
In the special case when $R \subseteq X \times Y$ is a function $f$, this definition says that 
\[
	 f(x\cdot\pi) = f(x)\cdot\pi \qquad \mbox{ for } x\in X,\ \pi\in G,
\]
where the action on the left is taken in $X$ and on the right in $Y$. The identity function on any $G$-set is equivariant, and the composition of two equivariant functions is again equivariant, therefore for any group $G$,  $G$-sets and equivariant functions form a category, called  $\GSet$.

If a singleton subset $\set{x}$ of a $G$-set is equivariant, we say that $x$ is an equivariant element of the $G$-set. In other words, an equivariant element is one that is preserved under the action of every element from $G$. Again in other words, an equivariant element is one that has a singleton orbit under the action of $G$.

\begin{exa}
In the equality symmetry, the only equivariant function from $\D$ to $\D$ is the identity; there are exactly two equivariant functions from $\D^2$ to $\D$ (the projections), and exactly one from $\D$ to $\D^2$ (the diagonal function $d\mapsto (d,d)$). Also, the mapping $(d,e)\mapsto\{d,e\}$ is the only equivariant function from $\D^{(2)}$ to $\choose{\D}{2}$. 

There is no equivariant function from $\choose{\D}{2}$ to $\D^{(2)}$. To see this, first note that if an equivariant function maps $\{d,e\}$ to $(b,c)$ then $b,c\in\{d,e\}$. Indeed if, say, $b\not\in\{d,e\}$ then the permutation $(b\ b')$ that swaps $b$ with a fresh $b'$, leaves $\{d,e\}$ intact in $\choose{\D}{2}$ but changes $(b,c)$ into $(b',c)$ in $\D^{(2)}$. Now, for any $d, e \in \D$, assume that a function maps $\{d,e\}$ to $(d,e)$ (the case of $(d,d)$ is similar). The uniquely induced equivariant relation
\[
\set{ (\set{d,e} \cdot \pi, (d,e) \cdot \pi) : \pi \in G}
\]
is not a function, since the permutation $(d\ e)$ that swaps $d$ and $e$ leaves $\{d,e\}$ intact in $\choose{\D}{2}$, but changes $(d,e)$ into $(e,d)$ in $\D^{(2)}$.
\end{exa}

\paragrafik{Languages} The classical notion of a language directly generalizes to the world of $G$-sets. An \emph{alphabet} is any orbit finite $G$-set $A$. Examples of alphabets in the symmetries mentioned so far include the set of data values $\D$, any finite set $\Sigma$, or a product $\Sigma \times \D$ where $\Sigma$ is finite.  When $A$ is an alphabet, the set of strings $A^*$ is treated as a $G$-set, with the point-wise action of $G$. A  $G$-language is any equivariant subset  $L \subseteq A^*$.

\begin{exa}
In the examples below assume $A = \D$. In the equality symmetry, exemplary $G$-languages are:
\[
\bigcup_{d \in \D} d \cdot \D^* \cdot d
\qquad \qquad
\bigcup_{d, e\in \D} (d \, e)^*
\qquad \qquad
\set{d_1 \ldots d_n : n \geq 0, d_i \neq d_j \text{ for } i \neq j}
\]
or palindromes over $\D$.
In the total order symmetry, all monotonic words
\[
\set{d_1 \ldots d_n : n \geq 0, d_1 < \ldots < d_n }
\]
is a $G$-language.
\end{exa}

\newcommand{\orbits}[2]{\mathit{orbits}_{#1}(#2)}

\section{G-automata} \label{sec:gautomata}

The notion of $G$-automaton, to be introduced now, is an obvious generalization of classical automata to $G$-sets.
The definition is exactly like the classical one, except that 
\begin{itemize}
\item the notion of finiteness is relaxed: \emph{orbit finite} sets are considered instead of \emph{finite} ones, and 
\item the components of the automaton, such as the initial and accepting states, or the transition relation, are required to be equivariant.
\end{itemize}
Our main observation in this section is that the Myhill-Nerode theorem may be lifted to the general setting of $G$-automata.
This is the first step in the program that we develop later in Sections~\ref{sec:nominalgautomata}--\ref{sec:others}.

For the rest of this section we fix some data symmetry $(\D,G)$.

	\begin{defi}\label{def:nondetGautomata}
		A \emph{nondeterministic  $G$-automaton} consists of 
		\begin{itemize}
			\item an orbit finite $G$-set $A$, called the input alphabet,
			\item a $G$-set $Q$, the set of states,
			\item equivariant subsets $I,F \subseteq Q$ of initial and accepting states,
			\item an equivariant transition relation
			\begin{align*}
				\delta \subseteq Q \times A \times Q.
			\end{align*}
		\end{itemize}
		We say that the automaton is orbit finite if the set of states $Q$ is so.
	\end{defi}
To define acceptance, we extend the single-step transition relation $\delta$ to the multi-step relation
\[
\delta^* \subseteq Q \times A^* \times Q
\]
in the usual way. A word $w \in A^*$ is accepted by an automaton if 
$(q_I, w, q_F) \in \delta^*$ for some initial state $q_I$ and accepting state $q_F$.
Note that $\delta^*$ is equivariant, similarly as $I$ and $F$, and thus the set of words accepted by a $G$-automaton is a $G$-language.

\subsection{Deterministic $G$-automata} \label{sec:detgautomata}

From now on, unless stated otherwise, we only consider \emph{deterministic $G$-automata}, the special case of a nondeterministic ones
where the transition relation is a function
\begin{align*}
	\delta:  Q \times A \to Q,
\end{align*}
and where the set of initial states is a singleton $\set{q_I}$.
A deterministic $G$-automaton is called \emph{reachable} if every state is equal to  $\delta^*(q_I, w)$ for some $w \in A^*$.

\begin{exa}
\label{ex:lang} In this example assume the equality symmetry $G = \permgrp(\D)$.
	We describe a deterministic $G$-automaton recognizing the language
	\begin{align*}
		\set{def : f \in \set{d,e}}.
	\end{align*}
	 Its states are $\bot,\top$, as well as tuples of data values of size at most two:
	\begin{align*}
		Q = \set{\top,\bot,\epsilon}\cup \D \cup \D^2.
	\end{align*}
	The state space $Q$ has six orbits: three singleton orbits
	\begin{align*}
		\set{\bot},\ \set{\top},\ \set{\epsilon},
	\end{align*}
	and three infinite orbits
	\begin{align*}
		\set{d : d \in \D},\ \set{(d,d) : d \in D},\ \set{(d,e) : d \neq e \in \D}.
	\end{align*}
	with the obvious pointwise action of $G$ as in Example~\ref{ex:gsets}.
	
The idea is that the automaton, when reading the first two letters of its input, simply stores them in its state. Then, after the third letter, it has state $\top$ or $\bot$ depending on whether its input belongs to $L$ or not. Formally, the transition function $\delta:Q\times\D\to Q$ is defined by cases:
\begin{align*}
	\delta(\epsilon,d) &= d \\
	\delta(d,e) &= (d,e) \\
	\delta((d,e),f) &= \left\{\begin{array}{rl}
					\top & \mbox{if } f\in\{d,e\} \\
					\bot & \mbox{otherwise}
				  \end{array}\right. \\
	\delta(\top,d) &= \delta(\bot,d) = \bot
\end{align*}
This function is easily seen to be equivariant. The only accepting state is $\top$, and $\epsilon$ is the initial one.
	
\end{exa}

\begin{exa}
\label{ex:syntaut}
	Consider the same group $G$ and the same language as in the previous example.  We describe a different automaton for the language. Its states are $\bot,\top$, as well as nonempty {\em sets} of data values of size at most two:
\[
	Y = \set{\top,\bot,\epsilon}\cup \D \cup \choose{\D}{1} \cup 
            \choose{\D}{2}.
\]	
(In the above, $\choose {\D} k$ refers to subsets of $\D$ that have size exactly $k$.)
	  One can give an equivariant transition function on these states by analogy to the above example, so that the resulting automaton recognizes the same language. The idea is that a state  $d \in \D$ represents a word $d$ of one letter, and a state  $\set d \in \choose{\D} 1$  represents a word $dd$ of two letters, where the letters happen to be equal. Compared to the automaton from the previous example, the change is that instead of the orbit 
	\[
	O_1 = \set{(d,e) : d \neq e \in \D}
	\]
	we have an orbit
	\[
	O_2 = \set{\set{d,e}: d \neq e \in \D}.
	\]
	In particular, both automata have six orbits of states. However, the new automaton is smaller in the following sense: there is an equivariant surjection from $O_1$ to $O_2$, but there is no equivariant function from $O_2$ to $O_1$. Intuitively, the new automaton is more abstract in that it ignores the order of the two data stored in memory.
\end{exa}

\paragrafik{Categorical perspective}
Viewing an element of $Q$ as a function from a singleton set $1=\{\star\}$ to $Q$ and a subset of $Q$ as a function from $Q$ to a two-element set $2$, one can depict an automaton using  a diagram:
\begin{equation}\label{dgm:automaton}
\vcenter{\xymatrix @R=1pc {
& 1\ar[d]^{\iota} \\
Q\times A\ar[r]_{\delta} & Q\ar[r]_{\alpha} & 2.
}}
\end{equation}
In the categorical approach to automata theory (see e.g.~\cite{Adamek:1990:AAC:575450} and references therein), it is standard to define various kinds of sequential automata by instantiating this diagram in suitable categories. In this paper, we study the case of the category  $\GSet$; this amounts to interpreting all objects in~\eqref{dgm:automaton} as $G$-sets and arrows as equivariant functions. We consider the trivial $G$-action on the sets $1$ and $2$. This means that the initial state is a singleton orbit, and the set of accepting states is a union of orbits.

%\begin{fact}
%\label{fact:automatacat}
%Deterministic $G$-automata are instances of the diagram~(\ref{dgm:automaton}) in $\GSet$, with $A$ assumed to be orbit finite. 
%\end{fact}
%The set $X$ is called the set of \emph{configurations}\footnote{Why do we use the name configurations instead of states? In the sequel, some automata will be presented by  giving a finite state space $Q$ and a set  of registers $R$. Then, a configuration of the automaton will consist of a control state $q  \in Q$ together with a valuation  $v: R \to \D$  mapping  registers to data values. The idea is that configurations will be orbit finite, and states will be finite in the usual sense.
%}, and the set $A$ is called the \emph{input alphabet}.
Just as $Q$ and $A$ are typically assumed to be finite sets in the classical case, we will typically require them 
to be orbit finite. 
This again follows from abstract categorical principles, as orbit finite $G$-sets are exactly finitely presentable objects in $\GSet$, just as finite sets are finitely presentable in the category $\Set$ of sets and functions (see e.g.~\cite{AR94} for information on locally finitely presentable categories).
%An automaton itself is called orbit finite if both $X$ and $A$ are so.

%The transition function $\delta : X \times A \to X$ of an automaton extends from single letters to arbitrary words:
%\begin{align*}
%	\delta^* : X \times A^* \to X.
%\end{align*}
%The  resulting function is equivariant.

%According to what was said above, we choose orbit finite $G$-sets as playing the role of classical finite sets.
We note, however, that the Cartesian product of two orbit finite $G$-sets is not always orbit finite.
A counterexample, in the integer symmetry, has been provided in Example~\ref{ex:orbits}.
In particular, even if both $A$ and $Q$ are orbit finite, the domain $Q \times A$ of the transition function of a $G$-automaton is not always orbit finite.
This inconvenience will be avoided when we restrict to \Fra symmetries in Section~\ref{sec:fraisse}.

\subsection{Myhill-Nerode Theorem}
%\paragrafik{Syntactic automaton} 
The Myhill-Nerode equivalence relation makes sense for any alphabet $A$, including infinite alphabets. That is, we consider two words $w,w' \in A^*$ to be equivalent with respect to a language $L \subseteq A^*$, denoted by $w \equiv_L w'$, if 
\begin{align*}
 	wv \in L \iff w'v \in L \qquad \mbox{for every $v \in A^*$}.
\end{align*}
\begin{lem} \label{lem:MNequivariant}
If $L$ is equivariant then $\equiv_L$ is equivariant too.
\end{lem}
\begin{proof}
We need to show:
\begin{equation}
	\label{eq:mn-class-to-class}
	w \equiv_L w' \quad\mbox{implies}\quad w \cdot \pi \equiv_L w' \cdot \pi.
\end{equation}
Indeed, to prove the above observation, suppose that $w \equiv_L w'$. By unraveling the definition of $\equiv_L$, we need to show that, for all $v \in A^*$, the following equivalence holds.
\begin{align*}
	(w \cdot \pi) v \in L \iff (w' \cdot \pi) v \in L
\end{align*}
By acting on both sides by $\pi^{-1}$, this is equivalent to 
\begin{align*}
	((w \cdot \pi) v)\cdot \pi^{-1} \in L\cdot \pi^{-1} \iff ((w' \cdot \pi) v)\cdot \pi^{-1} \in L\cdot \pi^{-1} 
\end{align*}
By equivariance of $L$, this is equivalent to
\begin{align*}
	((w \cdot \pi) v)\cdot \pi^{-1} \in L \iff ((w' \cdot \pi) v)\cdot \pi^{-1} \in L
\end{align*}
By equivariance of concatenation in $A^*$, this is equivalent to 
\begin{align*}
	w (v\cdot \pi^{-1}) \in L \iff w' (v\cdot \pi^{-1}) \in L
\end{align*}
The above is implied by $w \equiv_L w'$, which completes the proof of~\eqref{eq:mn-class-to-class}. 
\end{proof}
Below we will use the property that the quotient of a $G$-set by an equivariant equivalence relation has a
natural structure of $G$-set:
\begin{lem} \label{lem:equivariantquotient}
Let $X$ be a $G$-set and let $R \subseteq X \times X$ be an equivalence relation that is equivariant.
Then the quotient $X / R$ is a $G$-set, under the action
\[
	[x]_R \cdot \pi = [x \cdot \pi]_R
\]
of $G$, and the abstraction mapping 
\[
x \mapsto [x]_R \ : \ X \to X/R
\]
is an equivariant function.
\end{lem}
\begin{proof}
Relying on equivariance of $R$, both well-definedness of the action of $G$, as well as
equivariance of the abstraction mapping, are routinely checked.
\end{proof}
As usual, the equivalence $\equiv_L$ is a congruence with respect to appending new letters, i.e.~if $w \equiv_L w'$ then $wa \equiv_L w'a$ holds for every letter $a \in A$.
Thus one can define a transition function on equivalence classes
\begin{align*}
	\delta_L : \quad A^*/\equiv_L \ \ \times \ \  A \quad \to \quad  A^*/\equiv_L
\end{align*}
such that:
\begin{equation}\label{eq:synt-delta-def}
	\delta_L ( [w]_{\equiv_L}, a )=  [w\, a]_{\equiv_L}.
\end{equation}

 If $A$ is a $G$-set and $L$ is a $G$-language then $\equiv_L$ is an equivariant relation on $A^*$. We call it  the \emph{syntactic congruence of $L$}. 

Suppose that $A$ is orbit finite and  $L\subseteq A^*$ is a $G$-language. We define the \emph{syntactic automaton} of $L$ as follows: its states are equivalence classes of $A^*$ under Myhill-Nerode equivalence $\equiv_L$, the transition function is $\delta_L$, its initial state is the equivalence class of the empty word $\emptyword$, and accepting states are equivalence classes of the words in $L$.  

\begin{lem}\label{lemma:syntaut}
	The  syntactic automaton of a $G$-language is a reachable deterministic $G$-automaton.
\end{lem}
%\begin{proof}
%The action of $G$ on the set $A^*/\equiv_L$ is defined by 
%\begin{align} \label{eq:action}
%[w]_{\equiv_L} \cdot \pi = [w \cdot \pi]_{\equiv_L}.
%\end{align}
%It is well defined for a $G$-language $L$, as $w \equiv_L v$ implies $w \cdot \pi \equiv_L v \cdot \pi$.
%We show that $\delta$ is an equivariant function, i.e.
%\[
%\delta([w]_{\equiv_L} \cdot \pi, a \cdot \pi) = \delta([w]_{\equiv_L}, a) \cdot \pi
%\]
%for any $w$, $a$ and $\pi$, by rewriting both sides to $[wa \cdot \pi]_{\equiv_L}$, using equations~\eqref{eq:delta} and~\eqref%{eq:action}.
%Similarly one shows that the initial and accepting configurations are closed under the action of $G$.
%Finally we note that the automaton is reachable as $[w]_{\equiv_L} = \delta^*([\emptyword]_{\equiv_L}, w)$.
%\end{proof}
\proof
Note that we do not claim the syntactic automaton to be orbit finite.

By Lemma~\ref{lem:MNequivariant} the congruence $\equiv_L$ is equivariant, and thus 
Lemma~\ref{lem:equivariantquotient} applies. Thus we can
define an action of $G$ on equivalence classes of $\equiv_L$ by
\begin{equation}
	\label{eq:mn-definition}
	[w]_{\equiv_L} \cdot \pi = [w \cdot \pi]_{\equiv_L}.
\end{equation}	
So far, we have defined the structure of a $G$-set on the state space of the syntactic automaton. 
%Why is this $G$-set nominal? By~\eqref{eq:mn-definition}, the function $w \mapsto [w]_{\equiv_L}$ is an equivariant function from the nominal set $A^*$ to the state space. The image of a nominal set under an equivariant function is also a nominal set.
%
To complete the proof of the lemma, we need to show that the various components of the syntactic automaton are equivariant.
It is easy to see that  the initial state is a singleton orbit:
\begin{align*}
	[\epsilon]_{\equiv_L}\cdot \pi  \stackrel  {\eqref{eq:mn-definition}}= [\epsilon \cdot \pi]_{\equiv_L} = [\epsilon]_{\equiv_L}.
\end{align*}
By equivariance of $L$, the set of final states is also equivariant:
\begin{align*}
	[w]_{\equiv_L} \in F \iff w \in L \iff w \cdot \pi \in L \iff [w \cdot \pi]_{\equiv_L} \in F \iff [w]_{\equiv_L}  \cdot \pi \in F.
\end{align*}
Finally, the transition function in the syntactic automaton is equivariant:
\begin{align*}
	\delta_L([w]_{\equiv_L},a)\cdot \pi \stackrel
              {\eqref{eq:synt-delta-def}}= [w\cdot a]_{\equiv_L}\cdot
              \pi \stackrel  {\eqref{eq:mn-definition}}= [(w \cdot
                \pi)\cdot (a \cdot \pi)]_{\equiv_L} \stackrel
                  {\eqref{eq:synt-delta-def}}= \delta_L([w]_{\equiv_L}
                  \cdot \pi,a \cdot \pi).\rlap{\hbox to 30 pt{\hfill\qEd}}
\end{align*}

	For the language in Example~\ref{ex:lang}, the syntactic automaton is the one in Example~\ref{ex:syntaut}, and not the one in Example~\ref{ex:lang}.

\paragrafik{Homomorphisms of automata} 
%Suppose that $\Aa$ and $\Bb$ are $G$-automata over $A$. A homomorphism from $\Aa$ to $\Bb$
%is an equivariant function from configurations of $\Aa$ to configurations of $\Bb$ that 
%preserves and reflects initial and accepting configurations 
%and commutes with the transition functions of $\Aa$ and $\Bb$.
%This is formally expressed by saying that the following diagram commutes:
%\begin{equation}\label{dgm:automatonmorphism}
%\vcenter{\xymatrix @R=1pc {
%& 1 \ar@{=}[ddr] \ar[d]_{\iota}  \\
%X\times A\ar[r]_{\delta} \ar[rdd]_{f \times \text{id}_A} & X\ar[r]^{\ \ \ \ \alpha} \ar[rdd]_f & 2 \ar@{=}[ddr] \\
%&& 1 \ar[d]^{\kappa} \\
%& Y\times A\ar[r]_{\rho} & Y\ar[r]_{\beta} & 2 
%}}
%\end{equation}
%
Suppose that we have two deterministic $G$-automata
\begin{align*}
	\Aa = (Q,A,q_I,F,\delta) \qquad \Aa' = (Q',A,q'_I,F',\delta')
\end{align*}
over the same input alphabet $A$. An equivariant function
\begin{align*}
	f : Q \to Q'
\end{align*}
is called an automaton homomorphism if it maps $q_I$ to $q'_I$, maps $F$ to $F'$:
\[
	q\in F \mbox{ iff } f(q)\in F' \qquad \mbox{for every }q\in Q,
\]
and commutes with the transition functions $\delta$ and $\delta'$:
\begin{align*}
	f(\delta(q,a))=\delta'(f(q),a) \qquad \mbox{for every }q \in Q \mbox{ and }a \in A.
\end{align*}
It is easy to see that two automata related by a homomorphism recognize the same language.
% recognized languages are preserved under homomorphisms of automata.
If there is a surjective homomorphism from $\Aa$ to $\Aa'$ then we call $\Aa'$ a homomorphic image of $\Aa$.

\paragrafik{Myhill-Nerode theorem}
In Theorem~\ref{thm:MN} below we state an abstract counterpart of the Myhill-Nerode theorem for infinite alphabets.
The proof relies on Lemma~\ref{lemma:syntaut} and on the following fact:
\begin{lem}\label{lemma:homoimage}
	Let $L$ be a $G$-language. The syntactic automaton of $L$ is a homomorphic image of any reachable deterministic $G$-automaton that recognizes $L$.
\end{lem}
\begin{proof}
Consider a reachable deterministic $G$-automaton that recognizes $L$, over the alphabet $A$, 
with initial state $q_I$ and transition function $\delta$.
We claim that the mapping 
\[
\delta^*(q_I, w) \ \longmapsto \ [w]_{\equiv_L}, \quad \text{ for } w \in A^*,
\]
is a homomorphism. 
It is total as the automaton is reachable, and well defined as 
$\delta^*(q_I, w) = \delta^*(q_I, v)$ implies $w \equiv_L v$.
The mapping is easily shown equivariant using Lemmas~\ref{lem:MNequivariant} and~\ref{lem:equivariantquotient}.
It commutes with the transition functions by the very definition of the syntactic automaton.
The initial state $q_I$ is mapped to the initial one $[\emptyword]_{\equiv_L}$.
Finally, the accepting states are mapped to accepting states, as
$\delta^*(q_I, w)$ or $[w]_{\equiv_L}$ is accepting exactly when $w \in L$.
\end{proof}
\begin{thm}[Myhill-Nerode theorem for $G$-sets] \label{thm:MN}
Let $A$ be an orbit finite $G$-set, and let  $L \subseteq A^*$  be a $G$-language. The following conditions are equivalent:
	\begin{enumerate}
		\item the set of equivalence classes of Myhill-Nerode equivalence $\equiv_L$ is orbit finite;
		\item $L$ is recognized by a deterministic orbit finite $G$-automaton.
	\end{enumerate}
%A $G$-language $L$ is recognized by an orbit finite $G$-automaton iff $A^{\ast}/\equiv_L$ is orbit finite.
\end{thm}
\begin{proof}
The implication (1)$\implies$(2) follows by Lemma~\ref{lemma:syntaut}. For the opposite implication, 
we observe that if $L$ is recognized by a deterministic $G$-automaton $\Aa$ then without
loss of generality one may assume that $\Aa$ is reachable, and then use
Lemma~\ref{lemma:homoimage}.
\end{proof}

\section{Nominal $G$-sets}\label{sec:nominals}

The notion of $G$-automaton presented in Section~\ref{sec:gautomata} is quite abstract. When working with a model of computation, one expects it to have some kind of concrete presentation, e.g., in terms of control states and memory. Such a presentation makes it easier to understand what the automaton does, and is necessary to design algorithms that work with automata, e.g., minimization algorithms. Although we have defined some particular automata by finite means (e.g.~Example~\ref{ex:lang}), it is not clear how an arbitrary automaton can be presented.

One of the goals of this paper is to give a concrete presentation for orbit finite $G$-sets, equivariant functions and algebraic structures such as automata. This, however, cannot be done in full generality even for the equality symmetry (see Example~\ref{ex:symmetries}), for rather fundamental reasons:

 \begin{fact}\label{fact:uncountable}
For a countably infinite $\D$ and $G=\permgrp(\D)$, there are uncountably many non-isomorphic single-orbit $G$-sets.
\end{fact}
\begin{proof}
This proof is best deferred until Proposition~\ref{prop:uncountable}, after some basic representation machinery is introduced.
\end{proof}

Another problem with $G$-sets is that Cartesian product on them does not preserve orbit finiteness in general:

\begin{exa}
Consider $G=\permgrp(\D)$ for a countably infinite $\D$, and let $X\subseteq\P(\D)$ be the set of all those subsets of $\D$ that are neither finite nor cofinite. It is easy to see that $X$ is a single-orbit $G$-set. However, $X^2$ has infinitely many orbits. Indeed, for any $n\in\N$ one can choose $(C_n,D_n)\in X^2$ such that $|C_n\cap D_n|=n$, and pairs $(C_n,D_n)$ and $(C_m,D_m)$ are in different orbits of $X^2$ if $n\neq m$.
\end{exa}

Due to these difficulties, since the equality symmetry $G=\permgrp(\D)$ is one of the most important cases we want to consider, we need to restrict attention to some class of well-structured $G$-sets. To this end, we introduce the notion of a $G$-nominal set. Observe that so far, we have only used the group $G$, and we have ignored the fact that $G$ is a group acting on some data values $\D$. The definition of a $G$-nominal sets is where the data values start to play a role.

From now on, we focus on $G$-sets for groups arising from data symmetries. Consider a data symmetry $(\D,G)$ (cf.~Definition~\ref{def:symmetry}).

\begin{defi}\label{def:support}
A set $C\subseteq\D$ {\em supports} an element $x\in X$ if $x\cdot\pi=x$
for all $\pi\in G$ that act as identity on $C$. A $G$-set is \emph{nominal} in the symmetry $(\D,G)$ if every element of it has a finite support.
\end{defi}

Note that the definition of support mentions two group actions of $G$: an action on $X$, and the canonical one on $\D$. By abuse of notation, we usually leave the set of data values $\D$ implicit, and simply talk about nominal $G$-sets.

Nominal $G$-sets and equivariant functions between them form a category $\GNom$.

\begin{exa}\label{ex:wivwvw}
For any data symmetry, $\D$ is a nominal $G$-set, since every element $d\in\D$ is supported by $\{d\}\subseteq\D$. Similarly $\{d_1,\ldots,d_k\}$ supports $(d_1,\ldots,d_k)\in\D^k$, hence $\D^k$ is also a nominal $G$-set. The same works for $\D^*$, but not for $\D^\omega$ or $\P(\D)$ if $\D$ is infinite.

If $X,Y$ are nominal $G$-sets then so are the Cartesian product $X\times Y$ and the disjoint union $X+Y$. Indeed, if $C$ supports $x\in X$ and $D$ supports $y\in Y$ then $C\cup D$ supports $(x,y)\in X\times Y$, and also $C$ supports $x\in X+Y$ and $D$ supports $y\in X+Y$. A set $X$ equipped with the trivial $G$-action is always nominal, with every element supported by the empty set.
\end{exa}

\begin{exa}\label{ex:vonowv}
For the equality symmetry (see Examples~\ref{ex:symmetries}), nominal $G$-sets are exactly nominal sets introduced by Gabbay and Pitts~\cite{GP02}. Assuming $\D=\N$, the sets $\{0,1,2,3\}$ and its complement $\N\setminus\{0,1,2,3\}$,
considered as elements of $\P(\D)$, are both supported by $\{0,1,2,3\}$. In the equality symmetry, an element of $\P(\D)$ has finite support if and only if it is finite or cofinite. In particular, there are countably many finitely supported elements in ${\P}(\D)$.
\end{exa}

\begin{exa}\label{ex:vanofb}
Consider the total order symmetry, where $\D=\Q$, and the element $x\in \P(\Q)$ that is the union of two  intervals $[0;1] \cup [2;3)$.  It is easy to see that this element is supported by the set $\set{0,1,2,3}$. More generally, an element of $\P(\Q)$ has a finite support if and only if it is a finite Boolean combination of intervals. 
\end{exa}

\begin{exa}\label{ex:intsig}
	Consider the integer symmetry. If a translation $i \mapsto i + j$ preserves any single integer, then it is necessarily the identity. Therefore, any element of any set with an action of integers is supported by $\set{5}$ or $\set{8}$, etc. In the integer symmetry, all $G$-sets are nominal.
\end{exa}

Suppose that we change a symmetry $(\D,G)$ by keeping the set of data values $\D$, but considering a subgroup $H \leq G$. 
What happens to the nominal sets?  If $X$ is a $G$-set (and therefore also a $H$-set), then every $G$-support of $x\in X$ is also an $H$-support of $x$, therefore every nominal $G$-set is a nominal $H$-set.
On the other hand, under the smaller group $H$, more sets might become nominal (see Examples~\ref{ex:vonowv} and~\ref{ex:vanofb}). 

A basic property of equivariant functions is that they preserve supports:

\begin{lem}\label{lem:equiv-pres-supp}
For any equivariant $f:X\to Y$, $x\in X$ and $C\subseteq \D$, if $C$ supports $x$ then $C$ supports $f(x)$.
\end{lem}
\begin{proof}
For any $\pi\in G$, if $x\cdot\pi=x$ then $f(x)\cdot\pi = f(x\cdot\pi)=f(x)$.
\end{proof}
Similarly, action of the group preserves supports in the following sense:
\begin{lem} \label{lem:actionpreservessupport}
If $C$ supports $x$ then $\pi C$ supports $x \cdot \pi$, for any $\pi \in G$.
\end{lem}
\begin{proof}
Assume an arbitrary $\rho \in G$ to be the identity on $\pi C$.
Then $\pi \, \rho \, \pi^{-1}$ is the identity on $C$, and thus preserves $x$,
\[
x \cdot (\pi \, \rho \, \pi^{-1}) = x,
\]
from which we obtain:
\[
(x \cdot \pi) \cdot \rho = x \cdot \pi
\]
as required.
\end{proof}

The problem signified by Fact~\ref{fact:uncountable} disappears for nominal $G$-sets:

\begin{fact}\label{fact:countable}
For the equality symmetry $(\D,G)$, there are only countably many non-isomorphic single-orbit nominal $G$-sets.
\end{fact}
\begin{proof}
This will follow from the more general Corollary~\ref{cor:countable}.
\end{proof}

However, other problems persist and we shall not be able to distill a satisfactory representation of nominal $G$-sets and automata for arbitrary data symmetries. As a pathological example, consider the integer symmetry (see Example~\ref{ex:symmetries}).

\paragrafik{Integer pathologies}
As far as single-orbit nominal sets are concerned, the integer symmetry has a promisingly simple structure. As we mentioned in Example~\ref{ex:intsig}, all $G$-sets are nominal in this case. One example of a single-orbit $G$-set is $\Z$. Another example is the finite cyclic group $\Z_n$, for any nonzero $n \in \Nat$.  It is not difficult to see that 
every single-orbit nominal set in the integer symmetry is isomorphic either to $\Z$ or to some $\Z_n$. 

Equivariant functions between single-orbit sets are also simple.
If the domain in $\Z$, these are all translations, possibly modulo $n$ if the co-domain is $\Z_n$.
If the domain is $\Z_n$, the co-domain must be necessarily $\Z_m$ for $m$ a divisor of $n$.

The problems with the integer symmetry appear as soon as Cartesian products of nominal sets are considered. This has bad consequences for automata.
Suppose that we are interested in automata where the set of states is $\Z$ and the input alphabet is also $\Z$. Both sets are single-orbit and nominal, so these are among the simplest automata in the integer symmetry.
The transition function is any equivariant function
\[
	\delta : \Z \times \Z \to \Z.
\]
What functions $\delta$ can we expect? Suppose that $\delta$ is defined for arguments of the form $(0,i)$. Then, by equivariance, this definition extends uniquely to all arguments:
\begin{align*}
	\delta(i,j)= \delta((0,j-i)\cdot i)= \delta (0,j-i) + i .
\end{align*}
However, there is no restriction on the value of $\delta(0,i)$, call it $g(i)$. It is not difficult to show that for any function $g : \Z \to \Z$,
the function $\delta_g$ defined by
\[
	\delta_g(i,j) = g(j-i)+i
\]
is equivariant. In particular, there are uncountably many equivariant functions $\Z\times\Z \to \Z$.

Wishing to disregard the integer symmetry and other pathological cases, we shall require some desirable properties of data symmetries such as the existence of least supports.

\begin{defi}\label{def:admits-leastsupports}
A symmetry $(\D,G)$ \emph{admits least supports} if each element of every nominal $G$-set
has a least finite support with respect to set inclusion, or, equivalently, if finite supports of each element are closed under intersection.
\end{defi}

We shall study the existence of least supports in some detail in Section~\ref{sec:least-supp}. For now, we simply state some examples:

\begin{exa}\label{ex:leastsupp}
The equality symmetry and the total order symmetry both admit least supports. This will be proved as Corollaries~\ref{cor:leastsupportseq} and~\ref{cor:leastsupportsord}; for the equality symmetry, it was proved already in~\cite{GP02}.

The integer symmetry does not admit least supports, as is evident from Example~\ref{ex:intsig}.
\end{exa}

Later, in Section~\ref{sec:least-supp}, we shall restrict attention to those data symmetries that admit least supports and enjoy some other desirable properties, to achieve a finitary representation of nominal sets. 
For instance, Fact~\ref{fact:countable} holds for any data symmetry $(\D, G)$ that admits least supports, assumed that $\D$ itself is a countable set, as we show in Corollary~\ref{cor:countable}.
For now, however, we shall continue the study of nominal automata in an abstract setting, in the following Sections~\ref{sec:nominalgautomata}--\ref{sec:others}.

We conclude this section with a simple fact that gives a feeling of a kind of finitary representation we mean above.
Recall from Example~\ref{ex:gsets} that $\D^{(C)}$ is the set of all injective functions
$C \to \D$ that extend to a permutation from $G$.
\begin{lem}\label{lem:simple-representation-lemma}
	Every single orbit nominal set $X$ is an image, under an equivariant function, of $\D^{(C)}$, for some $C \subseteq_\text{fin} \D$. Moreover, if $(\D,G)$ admits least supports then the function may be chosen so that it preserves least supports.	
\end{lem}
\begin{proof}
       For the first part, choose any $x \in X$ and any finite support
	$C$    % = \set{d_1,\ldots,d_n}$ 
	of $x$. Because this is a support it follows that 
	\begin{align*}
		\pi|_{C} = \sigma|_{C} \quad \Rightarrow \quad x\cdot \pi = x \cdot \sigma \qquad \mbox{ for every }\pi,\sigma \in G.
	\end{align*}
	Therefore, the relation  defined by 
	\begin{align*}
%		f = \set{ (d_1,\ldots,d_n, x) \cdot \pi \in \D^{(n)} \times X  : \pi \in G}
                f = \set{ (\pi|_C, x \cdot \pi) : \pi \in G }
	\end{align*}
	is actually an equivariant function from $\D^{(C)}$ to $X$.  The function is clearly surjective, because every element of $X$ is of the form $x \cdot \pi$ for some $\pi \in G$.
	
	Assuming that $(\D,G)$ admits least supports, $C$ above may be chosen as the least support of $x$. It is easy to see that $\pi C$ is the least support both of $x\cdot\pi\in X$ (use Lemma~\ref{lem:actionpreservessupport} for $\pi^{-1}$) and of $\pi|_C\in\D^{(C)}$, which means that $f$ preserves least supports.
\end{proof}
In the equality symmetry every injective function $C \to \D$ extends to a permutation, thus we obtain:
\begin{cor}\label{cor:simple-representation-lemma}
	In the equality symmetry, every single orbit nominal set is an image  of $\D^{(n)}$, for some $n \in \Nat$, under an equivariant function that preserves least supports.	
\end{cor}

\section{Nominal $G$-automata} \label{sec:nominalgautomata}

This section is a continuation of our theory of $G$-automata initiated in Section~\ref{sec:gautomata}, 
but now we restrict attention to nominal $G$-sets.
We call a $G$-automaton \emph{nominal} if both the alphabet $A$ and the state space $Q$ are nominal $G$-sets.

Note that if $A$ is nominal then $A^*$ is so as well, as every finite word is supported by the union of supports of individual letters. Thus every $G$-language over a nominal alphabet is automatically nominal.

The restriction to nominal sets will have little or no impact on the expressive power of automata. In particular, in any data symmetry $(\D,G)$:

\begin{prop} \label{prop:reachablenominal}
In any reachable deterministic $G$-automaton over a nominal alphabet $A$, the $G$-set of states is nominal.
\end{prop}
\begin{proof}
Reachability of an automaton (see Definition~\ref{def:nondetGautomata} and the following paragraphs) means that the 
function  $w \mapsto \delta^*(q_I, w)$ is an equivariant function from the nominal set $A^*$ 
onto the state space of the automaton.
By Lemma~\ref{lem:equiv-pres-supp}, the image of a nominal set under an equivariant function is a nominal set.
%
%composite function 
%\[\xymatrix{
%	A^*\ar[rr]^-{w\mapsto(q_I,w)} & & Q\times A^*\ar[r]^{\delta^*} & Q
%}\]  
%is surjective. By Lemma~\ref{lem:equiv-pres-supp}, equivariant images of nominal $G$-sets are nominal.
\end{proof}

As before, for the rest of this section we fix some infinite set $\D$ and a group $G \leq \permgrp(\D)$.
The deterministic orbit finite nominal $G$-automata we call shortly \GDFA; similarly, the nondeterministic
orbit finite nominal $G$-automata we call \GNFA. 
% A language is called \emph{$G$-regular} if it is recognized by some \GDFA.

\subsection{Myhill-Nerode theorem revisited} Assume the alphabet $A$ is an orbit finite nominal $G$-set. 
By Proposition~\ref{prop:reachablenominal} every reachable deterministic automaton over $A$ is nominal. 
As a conclusion, the syntactic automaton is always nominal. 
Thus in condition (2) in Theorem~\ref{thm:MN} one may equivalently require that the automaton be nominal:
\begin{thm}[Myhill-Nerode theorem for nominal $G$-sets] \label{thm:MNrev}
Let $A$ be an orbit finite nominal $G$-set, and let  $L \subseteq A^*$  be a $G$-language. The following conditions are equivalent:
	\begin{enumerate}
		\item the set of equivalence classes of Myhill-Nerode equivalence $\equiv_L$ is orbit finite;
		\item $L$ is recognized by a \GDFA.
	\end{enumerate}
%A $G$-language $L$ is recognized by an orbit finite $G$-automaton iff $A^{\ast}/\equiv_L$ is orbit finite.
\end{thm}

\subsection{Nondeterministic G-automata}
In the sequel we investigate some basic properties of classical NFA, and verify which of them still hold 
for \GNFA.

\paragrafik{Determinization fails} In the world of nominal sets, one cannot in general determinize finite automata. One reason is that complementation fails for nondeterministic  finite automata. Perhaps  a  more suggestive explanation is that the powerset of an orbit finite set can have infinitely many orbits, as illustrated in Example~\ref{ex:orbits} in the case of the equality symmetry. 
%The set $\D$ of data values has one orbit. However, every finite subset of $\D$ is nominal. Because finite subsets of different cardinalities are in different orbits, it follows that there are infinitely many orbits in the nominal powerset of $\D$. 
This means that applying the subset construction to a nominal nondeterministic finite automaton yields a nominal deterministic automaton, but not necessarily one with an orbit finite state space.

\paragrafik{Elimination of $\emptyword$-transitions}
Consider nominal $G$-automata as in Definition~\ref{def:nondetGautomata}, 
but which also have additionally $\emptyword$-transitions, described by an equivariant relation
	\begin{align*}
		\delta_\emptyword \subseteq Q \times Q.
	\end{align*}

\begin{lem} \label{lem:epsilon-transitions}
The expressive power of \GNFA\  is not changed if $\emptyword$-transitions are allowed.
\end{lem}
\begin{proof}
	The standard proof works.
	After eliminating $\emptyword$-transitions, we should have transitions of the form $(p,a,q) \in Q \times A \times Q$ such that 
	\begin{align*}
		(p_1,p_2),\ldots,(p_{n-1},p_n) \in \delta_\emptyword \qquad (p_n,a,q_1) \in \delta \qquad  (q_1,q_2),\ldots,(q_{m-1},q_m) \in \delta_\emptyword
	\end{align*}
	holds 	for some 
	\begin{align*}
		p_1,\ldots,p_n, q_1,\ldots,q_m \in Q \qquad p_1 = p \qquad q_m=q.
	\end{align*}
	It is not difficult to see that the new set of transitions is equivariant.
\end{proof}

\paragrafik{Union and intersection}
	It is easy to see that languages recognized by \GNFA\ are closed under union (because orbit finite sets are closed under disjoint union)
	 and concatenation (disjoint union again, and using Lemma~\ref{lem:epsilon-transitions}). They also contain all orbit finite subsets of $A^*$.   This raises the question of regular expressions and a Kleene Theorem, but we do not discuss these issues in this paper.
	
Closure under intersection is a bit more subtle, as it does not hold in an arbitrary symmetry.
The essential reason is that orbit finite nominal sets are not stable under Cartesian product, as shown 
in Example~\ref{ex:orbits} in the case of the integer symmetry.
However, if one restricts to well-behaved symmetries only, as we do in Sections~\ref{sec:least-supp}-\ref{sec:Fraisseautomata},
the closure under products is recovered, and, as a consequence, the closure of \GNFA\ under intersection
is recovered as well.

\paragrafik{Complementation}
For closure under complementation, the situation is much worse, as the closure fails essentially in every symmetry.
The proof below works for the equality symmetry, but with minor changes it can be adapted to other symmetries.
	\begin{lem}
		In the equality symmetry, languages recognized by \GNFA\ are not closed under complementation.
	\end{lem}
	\begin{proof}\label{lem:no-complementation-for-nGa}
		Anticipating Section~\ref{sec:register-models}, we follow that same lines as the proof that finite memory automata of Francez and Kaminski are not closed under complementation.
		Consider the words over $\D$  which contain some data value twice:
		\begin{align*}
			L = \bigcup_{d \in \D} \D^* \cdot d \cdot  \D^* \cdot d \cdot  \D^*.
		\end{align*}
		The complement of this language is the set of words where all letters are distinct. Suppose that the complement of  $L$ is recognized by a \GNFA\ $\Aa$, with states $Q$ and transitions~$\delta$. For each $q \in Q$, let $C_q \subseteq \D$ be some chosen finite % the least 
support %\footnote{We know that the equality symmetry admits least supports (cf.~\cite{GP02} or Section~\ref{sec:least-supp} of this paper), but 
%                 in this proof we do not use this property.}
of $q$. By Lemma~\ref{lem:actionpreservessupport}, the sets $C_q$ may be chosen so that 
the size of $C_q$ depends only on the orbit of $q$, and therefore
		\begin{align*}
			 \max_{q \in Q}  |C_q|
		\end{align*}
		is a finite number, since there are finitely many orbits in $Q$. Choose $n \in \Nat$ to be bigger than this finite number.
		Consider a word 
		\begin{align*}
			d_1 \cdots d_{2n} \not \in  L.
		\end{align*}
			This word should be accepted by  $\Aa$, so there should be an accepting run
		\begin{align*}
			q_0,\ldots,q_{2n}  \qquad\mbox{such that $(q_{i-1},d_i,q_i) \in \delta$ for all $i \in \set{1,\ldots,2n}$}.
		\end{align*}
		Because the least support $C_{q_{n}}$ of $q_{n}$  has fewer than $n$ data values it follows that 
		\begin{align*}
			d_i,d_j \not \in C_{q_{n}} \qquad\mbox{ for some $i \in \set{1,\ldots,n}$ and some $j \in \set{n+1,\ldots,2n}$}.
		\end{align*}
		Let $\pi$  be the transposition which swaps $d_i$ and $d_j$. By equivariance of the transition relation, we see that the sequence
		\begin{align*}
			q_0 \cdot \pi,\ldots, q_n \cdot \pi
		\end{align*}
		is a run over the prefix 
		\begin{align*}
			(d_1 \cdots d_n) \cdot \pi.
		\end{align*}
		Because $\pi$ does not move the support of $q_n$, it follows that  $q_n \cdot \pi = q_n$.  Therefore, the sequence
		\begin{align*}
			q_0 \cdot \pi,\ldots, q_n \cdot \pi, q_{n+1},\ldots,q_{2n}
		\end{align*}
		is an accepting  run over the word 
		\begin{align*}
			((d_1 \cdots d_n) \cdot \pi) \cdot  d_{n+1} \cdots d_{2n}.
		\end{align*}
		However, the above word contains the data value $d_j$ twice, so it should be rejected by $\Aa$. 
		\end{proof}

\section{Relationship with finite memory automata}
\label{sec:register-models}
In this section, we take a detour from the discussion of automata theory in general symmetries, and we discuss the special case of the equality symmetry $(\D, G)$. In this case, for alphabets of a special form, the abstract model of nominal finite automata coincides with an existing automaton model, namely the finite memory automata of Francez and Kaminski. A connection between finite memory automata and nominal sets was first made in~\cite{ciancia-tuosto}, in the related framework of named sets and history-dependent automata. However, no comparison of the expressive power of automata was considered there.

\subsection{Finite memory automata}
\label{sec:finite-memory-automata}
We begin by defining \emph{finite memory automata}~\cite{FK94}, known also under the name \emph{register automata}~\cite{DL09}.

\paragrafik{Partial data tuples} 
Consider a finite set $N$ of names. A partial data tuple over $N$ is a partial function from $N$ to $\D$. We write $(\D \cup \bot)^N$ for the set of partial data tuples. 
An equality constraint over $N$ is an element
\begin{align*}
	(r,\tau,r') \in N \times \set{=,\neq} \times N
\end{align*}
We say a partial tuple $t$ satisfies the constraint if $t(r)$ is defined, and $t(r')$ is defined, and their data values are related by $\tau$.  For instance, the completely undefined tuple is the  unique  partial tuple that satisfies no constraints.  

\begin{lem}\label{lem:bool-comb-equality-constraints}
	Every equivariant subset of $(\D \cup \bot)^N$ is  equivalent to a boolean combination of equality constraints.
\end{lem}
\begin{proof}
Fix an arbitrary orbit of $(\D \cup \bot)^N$ and an arbitrary element $x$ of the orbit. 
Consider the set of equality constraints satisfied by $x$. A crucial but easy observation is that 
precisely the same constraints are satisfied by all other elements of the orbit.
On the other side, any two tuples that satisfy the same equality constraints are related by some permutation $\pi$.
Thus the orbit is equivalent to a conjunction of equality constraints. 
\end{proof}
There are only finitely many equality constraints,  as long as $N$ is finite, thus 
by the above lemma $(D \cup \bot)^N$ is an orbit finite nominal set.

\newcommand{\afin}{A_{\mathrm{fin}}}
\begin{defi}\label{def:finite-memory-automaton}
	A \emph{nondeterministic finite memory automaton} consists of 
	\begin{itemize}
		\item a finite set $\afin$ of input labels;
		\item a finite set $C$ of control states;
		\item a finite set $N$ of register names;
		\item sets of initial  $I \subseteq C$ and final $F \subseteq C$ control states;
		\item a transition relation, which is a subset of 
		\begin{align*}
			\delta \subseteq C \times \afin \times bool(\Phi) \times C
		\end{align*}
		where $\Phi$ is the set of equality constraints over the following set of names:
		\begin{align*}
			N'  = \qquad  \set{\mathrm{before}}\times N \quad \cup \quad \set{\mathrm{input}} \quad \cup \quad \set{\mathrm{after}}\times N
		\end{align*}
		and $bool(\Phi)$ stands for the boolean combinations of constraints from $\Phi$.
	\end{itemize}
\end{defi}\smallskip

\noindent Such an automaton $\Aa$ is used to accept or reject words over the alphabet $\afin \times \D$, and  works as follows. After reading a prefix of the input word, the configuration of the automaton consists  of a control state from $C$ together with a partial valuation from registers to data values. In other words, a configuration is an element of the set
\begin{align*} % \label{eq:confsetFMA}
Q_\Aa =	C \times (\D \cup \bot)^N.
\end{align*}
Initial configurations are the ones of the form
\[
	(c,\bot,\ldots,\bot)\in Q_\Aa
\]
where $c\in I$; note that there are only finitely many of them.
Suppose that the automaton is in a configuration
\begin{align*}
	(c,d_1,\ldots,d_n) \in Q_\Aa
\end{align*} 
and that it reads an input letter $(a,d) \in A$. The automaton can nondeterministically choose any new configuration
\begin{align*}
	(c',d'_1,\ldots,d'_n) \in Q_\Aa
\end{align*}
provided that there is a transition
\begin{align*}
	(c,a,\phi,c') \in \delta
\end{align*}
such that the partial tuple 
\begin{align*}
	(d_1,\ldots,d_n,d,d'_1,\ldots,d'_n),
\end{align*}
interpreted as a partial tuple over $N'$, satisfies the boolean combination of equality constraints given by $\phi$.

\begin{lem}\label{lem:simple-equivalence-with-Kaminski} Consider  an alphabet $A = \afin \times \D$, where $\afin$ is a finite set.  Then the following conditions are equivalent for every language $L \subseteq A^*$:
	\begin{enumerate}
			\item $L$ is recognized by a finite memory automaton.
		\item $L$ is recognized by a \GNFA , where
		\begin{itemize}
			\item The state space is $C \times (\D \cup \bot)^n$ for some finite set $C$ and $n \in \Nat$, 
			\item There are finitely many initial states.
		\end{itemize}
	\end{enumerate}
\end{lem}
\begin{proof}
	The  implication from (1) to (2) follows immediately from the definition; states of the \GNFA{} correspond to configurations in the finite memory automaton. Note that $\bot$ is a singleton orbit in $\D \cup \bot$. 
	For the converse implication, we use Lemma~\ref{lem:bool-comb-equality-constraints}.  The assumption on initial states guarantees that every initial state is of  the form $(c,\bot^n)$ for some $c \in C$.
\end{proof}

\subsection{Equivalence for nondeterministic automata}
\label{sec:nNFA=Kaminski}
In this section, we prove a stronger version of Lemma~\ref{lem:simple-equivalence-with-Kaminski}, namely:
\begin{thm}\label{thm:kaminski} Consider  an alphabet $A = \afin \times \D$, where $\afin$ is a finite set.  Then the following conditions are equivalent for every language $L \subseteq A^*$:
	\begin{enumerate}
		\item $L$ is recognized by a finite memory automaton.
		\item $L$ is recognized by a \GNFA.
	\end{enumerate}
\end{thm}
The implication from (1) to (2) has already been shown in Lemma~\ref{lem:simple-equivalence-with-Kaminski}. The   rest of Section~\ref{sec:nNFA=Kaminski} is devoted to the implication from (2) to (1).

%Using Lemma~\ref{lem:simple-representation-lemma} one easily observes:
\begin{cor}\label{cor:orbit finite-simple-representation}
	Every orbit finite nominal set is an image, under a partial  equivariant function $f$ that preserves least supports, of a  set of the form
	\begin{align*}
		I \times (\D \cup \bot)^n \qquad \mbox{for some finite set $I$ and $n \in \Nat$}.
	\end{align*}
\end{cor}
\begin{proof}
	Suppose that $X$ is a nominal set with $k$ orbits. 
Recall from Example~\ref{ex:gsets} that $\D^{(n)}$ is the set of non-repeating $n$-tuples of data values.
By Corollary~\ref{cor:simple-representation-lemma}, $X$ is an image of the disjoint union:
\begin{equation}
	\label{eq:typical-orbit finite-set}
	\coprod_{i \in \{1\ldots k\}} \D^{(n_i)}.
\end{equation}
%where $I$ is finite.
	Let $n$ be the maximal number among $\set{n_i}_{i \in \set{1 \ldots k}}$. It is not difficult to see that $\D^{(n_i)}$ is isomorphic to an orbit of $(\D \cup \bot)^n$. It follows that the disjoint union from~\eqref{eq:typical-orbit finite-set} is isomorphic to an equivariant subset of  $\set{1 \ldots k} \times (\D \cup \bot)^n$.
\end{proof}

We are now ready to prove Theorem~\ref{thm:kaminski}. Consider a \GNFA\ 
%nominal nondeterministic finite automaton
%\begin{align*}
	$\Aa = (Q,A,I,F,\delta)$
%\end{align*}
with $A = \afin \times \delta$ for some finite set $\afin$. We assume that there is only one initial state, call it $q_I$. Otherwise, we add a new initial state, call it $q_I$, with a trivial action
\begin{align*}
	q_I \cdot \pi = q_I,
\end{align*}
and extend the set of transitions by the equivariant set of triples of the form
\begin{align*}
	\set{(q_I,a,q) :  (p,a,q) \in \delta \mbox{ for some $p \in I$}}.
\end{align*}
Basing on Lemma~\ref{lem:simple-equivalence-with-Kaminski}, we only need to show that there is an equivalent
\GNFA\ with a single initial state, whose 
state space is $C \times (\D \cup \bot)^n$ for some finite set $C$ and $n \in \Nat$.
Apply Corollary~\ref{cor:orbit finite-simple-representation} to $Q$, yielding a partial surjective equivariant function
\begin{align*}
	f : \quad C \times (\D \cup \bot)^n  \quad \to \quad Q
\end{align*}
for some finite set $C$ and $n \in \Nat$. Because there is just one initial state, we may assume that 
\begin{align*}
	q_I = f(c_I,\bot^n) 
\end{align*}
for some $c_I \in C$. Define a \GNFA, call it $f^{-1}(\Aa)$, with states $C \times (\D \cup \bot)^n$, initial state $(c_I, \bot^n)$, final states $f^{-1}(F)$ and transitions $f^{-1}(\delta)$. It is easy to see that the automata $\Aa$ and $f^{-1}(\Aa)$ recognize the same language. This completes the proof of Theorem~\ref{thm:kaminski}.

\paragrafik{Local symmetry}
Although finite memory automata and \GNFA\ have the same expressive power, the latter model is arguably richer and has more structure. Indeed, in contrast to Lemma~\ref{lemma:syntaut}, syntactic automata of $G$-languages are not necessarily finite memory automata. An example is the automaton from Example~\ref{ex:syntaut}, which does not arise from any finite memory automaton. This is because \GNFA\ allow for a \emph{local symmetry}\footnote{The notion of local symmetry is introduced in its full generality in Section~\ref{sec:fraisse}.}, as illustrated in Example~\ref{ex:syntaut} where a \GNFA\  stores an unordered pair of data values instead of an ordered one; on the other hand finite memory automata do not allow any notion of local symmetry, or permutation, of registers. As a result, the Myhill-Nerode theorem fails, and finite memory automata do not minimize: the syntactic automaton is always a homomorphic image of a finite memory automaton, but it may not be isomorphic to one.

The importance of local symmetries for automata minimization was first noticed in the context of history-dependent automata, in~\cite{P99}.

\subsection{Equivalence for deterministic automata}
\label{sec:dNFA=Kaminski}

%We will need the following fact, known from the literature, which we actually reprove for a range of symmetries
%in Section~\ref{sec:least-supp}:
%\begin{prop}[\cite{GP02}]\label{def:admits-leastsupports}
%The equality symmetry \emph{admits least supports}: each element of every nominal set
%has the least support.
%\end{prop}

Recall that the set of configurations of a finite memory automaton $\Aa$ is $Q_\Aa =	C \times (\D \cup \bot)^N$.
The semantics of a nondeterministic finite memory automaton is given by a transition relation between configurations, 
being an equivariant subset
of $Q_\Aa \times (A_\text{fin} \times \D) \times Q_\Aa$. A finite memory automaton is called deterministic if this relation is actually a function
$Q_\Aa \times (A_\text{fin} \times \D) \to Q_\Aa$.

In this section, we  prove a deterministic variant of Theorem~\ref{thm:kaminski}:
\begin{thm}\label{thm:det-kaminski} Consider  an alphabet $A = \afin \times \D$, where $\afin$ is a finite set.  Then the following conditions are equivalent for every language $L \subseteq A^*$:
	\begin{enumerate}
		\item $L$ is recognized by a deterministic finite memory automaton.
		\item $L$ is recognized by a \GDFA.
	\end{enumerate}
\end{thm}

We do the same proof as for the nondeterministic automata. The only problem is that $f^{-1}(\delta) 
$ might not in general be deterministic. To solve this problem, we need one additional result: %two additional results. 

%Recall from Example~\ref{ex:leastsupp} that the equality symmetry admits least supports.
%Our first observation is that in Corollary~\ref{cor:orbit finite-simple-representation}
%there exists a function $f$ that is not just equivariant, but also surjective and preserves least supports. 
%(In general, an equivariant function might decrease the least support, see Lemma~\ref{lem:equiv-pres-supp}.) 
%The second additional result is stated in Lemma~\ref{lem:find-a-subset-function}.

\begin{lem}\label{lem:find-a-subset-function}
	Suppose that $f : X \to X'$ is a surjective equivariant function that preserves least supports. Then for every nominal set $A$, and every equivariant function 
	\begin{align*}
		\delta': X' \times A \to X'
	\end{align*}
	there exists a function
	\begin{align*}
		\delta : X \times A \to X
	\end{align*}
	such that the following diagram commutes
	\begin{equation}\label{dgm:commute-x1-x2}
	\vcenter{\xymatrix @R=1pc {
	X\times A\ar[r]^{\delta} \ar[d]_{f \times \text{Id}_A} & X \ar[d]^f  \\
	 X'\times A\ar[r]_{\delta'} & X'
	}}
	\end{equation}	
\end{lem}
\begin{proof}
	Let $(Y_i)_{i\in I}$ be the family of all orbits of the set $X \times A$.

	Consider  some $i \in I$.
	Pick a representative  $(x_i,a_i) \in Y_i$.
	In the diagram~\eqref{dgm:commute-x1-x2}, follow the $f \times \text{Id}_A$ arrow, and then $\delta'$, yielding an element
	\begin{align*}
		x'_i = \delta'(f(x_i),a_i).
	\end{align*}
	Because the above element is the  result of applying two equivariant functions, the least support of $x'_i$ is a subset of the least support of $(x_i,a_i)$.  Because the function $f$ is surjective, there must be some $y_i \in X$ such that
	\begin{align*}
		f(y_i) = x'_i=  \delta'(f(x_i),a_i).
	\end{align*}
	Because the function $f$ preserves least supports, the least support of $y_i$ is equal to the least support of $x'_i$, which is included in the least support of $(x_i,a_i)$. It follows that there is an equivariant function
	\begin{align*}
		\delta_i : Y_i \to X \qquad \mbox{ such that }\delta_i(x_i,a_i)=y_i.
	\end{align*}
	
	Do the construction above for all orbits $Y_i$, yielding functions $(\delta_i)_{i\in I}$. Define $\delta$ to be the union of these functions. We now prove that the diagram~\eqref{dgm:commute-x1-x2} commutes. 
	
	Pick some $(x,a) \in X \times A$.  Because the pairs $(x_i,a_i)$ for $i\in I$
	represent all orbits of $X \times A$, it follows that 
	\begin{align*}
		(x,a)=(x_i \cdot \pi,a_i \cdot \pi) \qquad\mbox{for some $i \in I$ and some $\pi \in G$}.
	\end{align*}
	Following the down-right path in the diagram~\eqref{dgm:commute-x1-x2} from $(x,a)$ yields
	\begin{align*}
		\delta'(f(x),a)= \delta'(f(x_i \cdot \pi), a_i \cdot \pi) = \delta'(f(x_i),a_i) \cdot \pi.
	\end{align*}
	Following the right-down path in the diagram~\eqref{dgm:commute-x1-x2} from $(x,a)$ yields
	\begin{align*}
		f(\delta(x,a)) = f(\delta_i(x),a)) = f(\delta_i(x_i \cdot \pi,a_i \cdot \pi))=f(\delta_i(x_i,a_i)) \cdot \pi = f(y_i) \cdot \pi, 
	\end{align*}
	which means that the diagram commutes because $f(y_i) = \delta'(f(x_i),a_i)$.

\end{proof}

We now prove Theorem~\ref{thm:det-kaminski}. 
Let $X'$ be the state space of the \GDFA\ from item (2), and let $\delta'$ be its transition function. Apply Corollary~\ref{cor:orbit finite-simple-representation}
 to $X'$, yielding a partial surjective equivariant function $f : X \to X'$
where
\begin{align*}
  X = C \times (\D \cup \bot)^n.
\end{align*}
Let $Y \subseteq X$ be the domain of $f$.
Because $f$ preserves least supports, we can apply Lemma~\ref{lem:find-a-subset-function} for $f$, yielding a transition function $\delta : Y \times A \to Y$. Extend $\delta$ to an equivariant function
$X \times A \to X$ in an arbitrary way. 
The rest of the proof is the same as in Theorem~\ref{thm:kaminski}, using Lemma~\ref{lem:simple-equivalence-with-Kaminski}.

%\subsection{Other symmetries} We have just proved that in the special case of the equality symmetry and alphabets of the form $\afin \times \D$, the abstract models of nominal (deterministic/nondeterministic) finite automata are equivalent to more concrete models, which have registers and instructions for manipulating them. 

%These results can be generalized to other symmetries. For instance, in the total order symmetry, one would use a model of automata with linearly ordered registers, where the transition relation is conditioned by the order of values of registers, like the automaton model considered in~\cite{FHL10}. 

	\section{Other models and perspectives} \label{sec:others}
	
In Sections~\ref{sec:gautomata} and~\ref{sec:nominalgautomata} we  defined and studied the nominal version of finite automata. The same approach could be pursued for a wide variety of computation models. For a simple example:

\begin{defi}\label{def:nominal-pda}
		A \emph{nominal pushdown automaton} $\Aa$ consists of 
		\begin{itemize}
			\item an input alphabet $A$, which is an orbit finite nominal set;
			\item a set of states $Q$, which is an orbit finite nominal set;
			\item a stack alphabet $\Gamma$, which is an orbit finite nominal set;
			\item an initial state $q_I \in Q$, which is equivariant;
			\item an initial stack symbol $\gamma_I \in \Gamma$, which is equivariant;
			\item a set of transitions
			\begin{align*}
				\delta \subseteq Q \times \Gamma \times (A \cup \epsilon)  \times Q \times \Gamma^*
			\end{align*}
			which is orbit finite and equivariant.
		\end{itemize}
	\end{defi}\smallskip

\noindent By analogy to classical pushdown automata, the condition that the set of transitions is orbit finite is to prohibit a set of rules which can push arbitrarily large words onto the stack in one step.
Assuming acceptance by empty stack, the acceptance by a pushdown automaton is defined exactly like in the classical case. 

\begin{exa}\label{example:palindrome-pda}
	For an orbit finite alphabet $A$, consider the language of even-length palindromes:
	\begin{align*}
		P = \set{a_1a_2 \cdots a_n a_n \cdots a_2a_1 : a_1,\ldots,a_n \in A} \subseteq A^*
	\end{align*}
	This language is recognized by a nominal pushdown automaton which works exactly the same way as the usual automaton for palindromes, with the only difference that the stack alphabet $\Gamma$ is now $A$. For instance, in the case when $A = \D$, the automaton keeps a stack of data values during its computation. The automaton has two control states: one for the first half of the input word, and one for the second half of the input  word.
	\end{exa}
	
\begin{exa}\label{example:mid-palindrome-pda}
		The automaton in Example~\ref{example:palindrome-pda} had two control states. In some cases, it might be useful to have a set $Q$ of control states that is orbit finite, but not finite. Consider for example the set of odd-length palindromes where the middle letter is equal to the first letter:
		\begin{align*}
			P' = \set{a_1a_2 \cdots a_na_1 a_n \cdots a_2 a_1 : a_1,\ldots,a_n \in A } \subseteq A^*.
		\end{align*}
		A natural automaton  recognizing this language would be similar to the automaton for palindromes, except that it would store the first letter $a_1$ in its control state. 
\end{exa}

Also the definition of a nominal context-free grammar is obtained from the standard definition by replacing `finite' with 'orbit finite', and requiring elements and subsets to be equivariant.
	
	\begin{defi}\label{def:cfg}
		A \emph{nominal context-free grammar} $\Gg$ consists of 
		\begin{itemize}
			\item an input alphabet $A$, which is an orbit finite nominal set;
			\item a set of nonterminals $\Nn$, which is an orbit finite nominal set;
			\item a starting nonterminal, which is equivariant;
			\item an orbit finite, equivariant set of productions
			\begin{align*}
				\Pp \subseteq \Nn \times (\Nn \cup A)^*
			\end{align*}
		\end{itemize}
		As usual, we assume that the sets $A$ and $\Nn$ are disjoint.
	\end{defi}

	\begin{exa}
		Consider the palindrome language $P$ from Example~\ref{example:palindrome-pda}. This language is generated by the following grammar with one nonterminal $N$:
		\begin{align*}
			N &\to aNa & \qquad \mbox{ for every }a \in A.\\
			N &\to \epsilon&
		\end{align*}
	\end{exa}
	
	\begin{exa}
		In the previous example, the grammar had just one nonterminal. Sometimes, it is useful to have an orbit finite, but infinite, set of nonterminals.
		Consider the language $P'$ from Example~\ref{example:mid-palindrome-pda}. For this language, we need a different set of nonterminals, with a starting nonterminal $N$ as well as one nonterminal $N_a$ for every $a \in A$. The rules of the grammar are:
		\begin{align*}
			N &\to aN_aa & &\qquad \mbox{ for every }a \in  A.\\
			N_a &\to bN_ab && \qquad \mbox{ for every }a,b \in  A.\\
			N_a &\to a & &\qquad \mbox{ for every }a \in  A.
		\end{align*}
		\end{exa}

\medskip
One can apply the same treatment to other classical definitions such as two-way automata, alternating automata (cf.~\cite{BBKL12}), Turing machines(cf.~\cite{BKLT13}), Petri nets, and so on.
In each case one has to be careful to see which of the classical constructions or equivalences work, and which of them fail. For example:

	\begin{itemize}
		\item Nominal pushdown automata are expressively equivalent to nominal context-free grammars. The proof is essentially the same as the standard proof for classical sets.
		\item Nominal two-way \GNFA\ (\GDFA) are more powerful than one-way \GNFA\ (\GDFA). For instance, the language
		\begin{align*}
			L= \set{d_1 \cdots d_n : \mbox{ $n \in \Nat$ and all the letters $d_1,\ldots,d_n$ are different}} \subseteq \D^*
		\end{align*}
		is recognized by a two-way \GDFA.
		\item Nominal alternating finite automata are more powerful than nominal nondeterministic finite automata. For instance, the language $L$ mentioned above is recognized by a nominal alternating finite automaton.  In the spirit of Section~\ref{sec:register-models}, one makes a connection between nominal alternating finite automata, and models of alternating register automata known in the literature~\cite{DL09, FHL10}. This connection is investigated in~\cite{BBKL12}.
		\item Determinization of Turing machines heavily depends on the data symmetry. In the equality symmetry, nondeterministic Turing machines are more powerful than deterministic ones, and P $\neq$ NP. In the total order symmetry, Turing machines determinize, and 
the P = NP question is equivalent to the classical one. These questions are investigated in detail in~\cite{BKLT13}.
	\end{itemize}
	
\noindent A general analysis of the types of reasoning allowed for nominal $G$-sets of various kinds is beyond the scope of this paper. One general rule is Pitts's {\em equivariance principle}:
\begin{quote}
Any function or relation that is defined from equivariant functions and relations using classical higher-order logic is itself equivariant~\cite{pitts-book}.
\end{quote}	
For example, the language recognized by an equivariant automaton is automatically equivariant.
	
In practice, various classical results in nominal sets fail either due to the fact that the finite powerset construction does not preserve orbit-finiteness (so, e.g., standard automata determination fails), or due to the failure of the axiom of choice, even in its orbit-finite form (see~\cite{BKLT13}). 

%%% Local Variables: 
%%% TeX-master: "main"
%%% End: 

\part{Finite representations of nominal sets and automata}

In the framework presented so far, automata and other models are generalized to infinite alphabets by reinterpreting their standard definitions, replacing finite sets by orbit-finite nominal sets and arbitrary relations and functions by equivariant ones. This is pleasantly simple, but not sufficient for a satisfactory treatment of the algorithmic aspect of automata theory.
For instance, for every deterministic finite automaton, one can minimize it, or test the emptiness of its recognized language, in polynomial time. To transport such results to the nominal case, one needs a finite representation of nominal data structures, amenable to effective manipulation.

% What is the nominal equivalent of these algorithms? How is the input represented? Is there a programming language? How is running time measured? What is polynomial time? A first step toward answering these questions is made in the paper~\cite{BBKL12}. The latter paper builds on finite representations of orbit finite sets, to be developed in the following sections of this paper.

We shall now provide finite representation results for nominal $G$-sets and equivariant functions. In Sections~\ref{sec:gset-repr}-\ref{sec:fraisse} we shall prove a sequence of progressively more concrete representations, under certain assumptions on the underlying data symmetry. An example application, shown in Section~\ref{sec:Fraisseautomata}, is a generalization of the development of Section~\ref{sec:register-models}, where a concrete understanding of deterministic and nondeterministic $G$-automata for the equality symmetry was provided. More applications can be found in~\cite{BBKL12}, where we define a programming language for manipulating orbit-finite nominal $G$-sets, with an implementation based on the representations presented here.

\section{$G$-set representation}\label{sec:gset-repr}
 
We begin with well-known results from group theory, regarding the structure of arbitrary $G$-sets for any group $G$, and we indicate why orbit finite $G$-sets cannot be presented by finite means in general. 

Important examples of $G$-sets are provided by subgroups of $G$ and their coset spaces. For a subgroup $H\leq G$, a (right) {\em coset} of $H$ is a set of the form 
\begin{align*}
	H\pi=\{\sigma\pi\mid\sigma\in H\}\subseteq G,
\end{align*} for some $\pi\in G$. Note that $H\pi=H\theta$ if and only if $\pi\theta^{-1}\in H$. Right cosets of $H$ define a partition of $G$, and the set of all such cosets is denoted $G/H^r$. 

We shall now show a well-known representation result for single-orbit $G$-sets as coset spaces of subgroups of $G$. 

\begin{defi}\label{def:subrepr}
A {\em subgroup representation} of a $G$-set is a subgroup $H\leq G$.
Its {\em semantics} is the set
\[
	\cossem{H} = G/H^r,
\] 
with a $G$-action defined by $(H\pi)\cdot\sigma=H(\pi\sigma)$ for any $H\pi\in G/H^r$ and $\sigma\in G$.
\end{defi}

The following two propositions are well known and their proofs completely standard; we include them here for completeness.

\begin{prop}\label{prop:subrepr} (1) For each $H\leq G$, $\cossem{H}$ is a single-orbit $G$-set.
(2) Every single-orbit $G$-set $X$ is isomorphic to some $\cossem{H}$.
\end{prop}
\begin{proof}
For (1), first check that the $G$-action on $\cossem{H}$ is well-defined under the choice of $\pi$; indeed, $H\pi=H\pi'$ implies $H(\pi\sigma)=H(\pi'\sigma)$. Further, every $H\pi,H\sigma\in\sem{H}$ are in the same orbit since $H\pi= H\sigma\cdot(\sigma^{-1}\pi)$. 

(2) is known in the literature as the {\em orbit-stabilizer theorem}. For any $x$ in a $G$-set $X$, the group 
\begin{align*}
	G_x=\{\pi\in G\mid x\cdot\pi=x\} \leq G
\end{align*}
is called the {\em stabilizer} of $x$.

To prove (2), put $H=G_x$ for any $x\in X$.
Define $f:X\to\cossem{G_x}$ by $f(x \cdot \pi) = G_x \pi$.
%$f(y)=G_x\pi$ for any $y\in X$, where $\pi\in G$ is such that $x\cdot\pi=y$. Such $\pi$ exists since $X$ has a single orbit. 
The function $f$ is well defined: if $x\cdot\pi=x\cdot\sigma$ then $\pi\sigma^{-1}\in G_x$, hence $G_x\pi=G_x\sigma$. It is easy to check that $f$ is equivariant. It is also a bijection. For injectivity, if $f(x \cdot \pi)=f(x \cdot \sigma)$, which means $G_x\pi=G_x\sigma$, 
% for $\pi,\sigma\in G$ such that $x\cdot\pi=y$ and $x\cdot\sigma=z$. 
then $\pi\sigma^{-1}\in G_x$, hence $x\cdot\sigma=(x\cdot\pi\sigma^{-1})\cdot\sigma=x\cdot\pi$. For surjectivity of $f$, for any $\pi\in G$ there is $f(x\cdot\pi)=G_x\pi$.
\end{proof}

Recall from group theory that subgroups $H,K\leq G$ are called {\em conjugate} if $K=\pi H\pi^{-1}$ for some $\pi\in G$.

\begin{prop}\label{prop:conjugate-repr}
For any $H,K\leq G$, $\cossem{H}$ and $\cossem{K}$ are isomorphic if and only if $H$ and $K$ are conjugate.
\end{prop}
\begin{proof}
For the {\em if} part, assume $K=\pi H\pi^{-1}$ and define
\[
	f(H\sigma)=K\pi\sigma.
\] 
This is well defined as a function from $\cossem{H}$ to $\cossem{K}$: if $H\sigma=H\theta$ then $\sigma\theta^{-1}\in H$, therefore $\pi\sigma\theta^{-1}\pi^{-1}=\pi\sigma(\pi\theta)^{-1}\in K$, hence $K\pi\sigma=K\pi\theta$. Moreover, $f$ is obviously equivariant by Definition~\ref{def:subrepr}, and the mapping $K\sigma\mapsto H\pi^{-1}\sigma$ is its inverse.

For the {\em only if} part, assume an equivariant isomorphism $f:\cossem{H}\to\cossem{K}$ and take any $\pi\in G$ such that $f(He)=K\pi$, for $e$ the neutral element of $G$. Now, for any $\sigma\in H$ there is 
\[
	K\pi\sigma = f(He)\cdot\sigma = f(H\sigma)=f(He)=K\pi
\]
hence $\pi\sigma\pi^{-1}\in K$; as a result, $H\leq \pi K\pi^{-1}$. For $f^{-1}$ the inverse of $f$, there is $f^{-1}(K\pi)=He$, therefore by equivariance, $f^{-1}(Ke)=H\pi^{-1}$ and by repeating the previous argument, $K\leq \pi^{-1}H\pi$, hence $\pi K\pi^{-1}\leq H$. As a result, $H=\pi K\pi^{-1}$ as required.
\end{proof}

The subgroup representation can be extended to a representation of equivariant functions from single orbit $G$-sets:

\begin{prop}\label{prop:subrepr-fun}
Let $X=\cossem{H}$ and let $Y$ be a $G$-set. Equivariant functions from $X$ to $Y$ are in bijective correspondence with elements $y\in Y$ for which $H\leq G_y$. 
\end{prop}
\begin{proof}
Given an equivariant function $f:X\to Y$, let $y$ be the image under $f$ of the coset $He \in X$. Equivariant functions can only increase stabilizers, so  $H=G_{He}\leq G_y$. On the other hand, given $y\in Y$, define a function $f: X \to Y$ by $f(H\pi)=y\cdot\pi$. This is well-defined if $H\leq G_y$; indeed, if $H\pi=H\sigma$ then $\pi\sigma^{-1}\in H\subseteq G_y$, hence $y\cdot\pi=y\cdot\sigma$.

It is easy to check that the two above constructions are mutually inverse.
\end{proof}

\begin{cor}\label{cor:subrepr-fun}
Equivariant functions from $X=\cossem{H}$ to $Y=\cossem{K}$ are in bijective correspondence with those cosets $K\pi$ for which $\pi H\subseteq K\pi$.
\end{cor}
\begin{proof}
By Proposition~\ref{prop:subrepr-fun} unfolding Definition~\ref{def:subrepr}. Notice that the stabilizer of $Ke\in\cossem{K}$ is $K$ itself, and the stabilizer of $K\pi$ is the conjugate subgroup $\pi^{-1}K\pi$. The condition $H\leq\pi^{-1}K\pi$ obtained from Proposition~\ref{prop:subrepr-fun} is equivalent to $\pi H\subseteq K\pi$. 
\end{proof}

Proposition~\ref{prop:subrepr} provides a way to represent single-orbit G-sets by subgroups. Together with Corollary~\ref{cor:subrepr-fun}, this representation can be rephrased concisely as an equivalence of two categories. Denote by $\GSetI$ the category of single-orbit $G$-sets and equivariant function between them.

\begin{thm}\label{thm:gsetrepr}
For any group $G$, $\GSetI$ is equivalent to a category with:
\begin{itemize}
\item as objects, subgroups $H\leq G$,
\item as morphisms from $H$ to $K$, cosets $K\pi$ such that $\pi H\subseteq K\pi$.
\end{itemize}
\end{thm}

We do not pursue the categorical formulation of G-sets in this paper, but we include this theorem to make a connection with related work such as~\cite{staton-thesis}, where formulated with essentially the same proof as above. 

Thanks to Proposition~\ref{prop:conjugate-repr}, one could refine Theorem~\ref{thm:gsetrepr} and represent single-orbit $G$-sets not by subgroups of $G$, but by conjugacy classes of those subgroups. For the sake of simplicity we choose not to do so.

The representation can be extended from single-orbit to arbitrary $G$-sets. To this end, note that the action of $G$ on a set $X$ acts independently on different orbits, and can be defined separately on each orbit. Formally, every $G$-set $X$ is isomorphic to the disjoint union of its orbits understood as single orbit $G$-sets. As a result, a $G$-set can be represented by a {\em family} of subgroups of $G$, and equivariant functions are represented as suitable families of functions.

The subgroup representation exhibits some structure in the world of $G$-sets and equivariant functions. At the same time, it implies that it is impossible to present all orbit finite $G$-sets by finite means, as we shall now demonstrate.

By Propositions~\ref{prop:subrepr} and~\ref{prop:conjugate-repr}, the following proposition proves Fact~\ref{fact:uncountable}.

\begin{prop}\label{prop:uncountable}
For a countably infinite $\D$ and $G=\permgrp(\D)$, there are uncountably many non-conjugate subgroups of $G$. 
\end{prop}
\begin{proof}
Fix an arbitrary family of pairwise-disjoint subsets $C_p\subseteq \D$, indexed by prime numbers $p$, such that $|C_p|=p$ for any $p$. Then, fix a family of permutations $\pi_p$, indexed also by prime numbers, 
such that each $\pi_p$ acts as identity on $\D\setminus C_p$, and as a permutation of order $p$ on $C_p$.
 For any subset $I$ of prime numbers, let the group $H_I\leq G$ be generated by the family $\set{ \pi_p : p \in I }$.
One easily observes that $H_I$ contains an element of a prime order $p$ if and only if $p\in I$.

%Fix an arbitrary family of subsets $C_p\subseteq \D$, indexed by prime numbers $p$, such that $|C_p|=p$ for any $p$. For any subset $I$ of prime numbers, let the group $H_I\leq G$ be generated by a fixed family of permutations $\pi_p$, also indexed by prime numbers, such that each $\pi_p$ is not an identity permutation, but acts as identity on $D\setminus C_p$. It is easy to see that each $\pi_p$ has order $p$ and, moreover, $H_I$ contains an element of a prime order $p$ if and only if $p\in I$.

On the other hand, is is easy to show that for conjugate subgroups $H,K\leq G$, if $H$ contains an element of some finite order, then $K$ contains an element of the same order. Therefore, if $I\neq I'$ then $H_I$ and $H_{I'}$ are not conjugate, and there are uncountably many different choices of $I$.
\end{proof}

\paragrafik{Open subgroups}
We shall now restrict the subgroup representation to nominal $G$-sets. 
For any $C\subseteq\D$, define $G_C\leq G$ by:
\begin{equation}\label{eqn:GC}
	G_C = \{\pi\in G \mid \pi(c)=c \mbox{ for all }c\in C\}.
\end{equation}
In other words, the subgroup $G_C$ is the intersection of all stabilizers $G_c$ for $c\in C$, in $\D$ considered as a $G$-set.

\begin{defi}\label{def:opengrp}
A subgroup $H\leq G$ is {\em open} if $G_C\leq H$ for some finite $C\subseteq \D$. If this is the case, we say that $C$ {\em supports} $H$.
\end{defi}

The name ``open'' is justified by considering $G\leq\permgrp(\D)$ as a topological group. This technique is well known, see e.g.~\cite{maclanemoerdijk} for an application in the context of sheaf theory closely related to nominal sets.
For any set $\D$, a set of permutations $G\subseteq\permgrp(\D)$ can be equipped with a topology 
with basis given by {\em $C$-neighborhoods} of all $\pi\in G$:
\begin{equation}\label{eq:ball}
	\mathcal{B}_C(\pi) = \{\sigma\in G\mid \sigma|_C=\pi|_C\}.
\end{equation}
It is not difficult to check that a subgroup $H\leq G$ is an open subset with respect to this topology if and only if it satisfies Definition~\ref{def:opengrp}.

Open subgroups of $G$ are linked to nominal $G$-sets via the following result.

\begin{prop}
A single-orbit $G$-set $\cossem{H}$ is nominal if and only if $H$ is open in $G$.
\end{prop}
\begin{proof}
Unfolding the definitions, it is easy to see that
in a $G$-set $X$, a subset $C\subseteq\D$ supports an element $x\in X$ if and only if $G_C\leq G_x$. Then use (the proof of) Proposition~\ref{prop:subrepr}(2).
\end{proof}

The above proof also implies that the notions of support in Definitions~\ref{def:support} and~\ref{def:opengrp} coincide along the representation function $\cossem{-}$. We shall use both notions as convenient.

It is now straightforward to restrict the subgroup representation of Definition~\ref{def:subrepr}: nominal $G$-sets are represented by open subgroups of $G$. The representation of equivariant functions from nominal sets remains as in Proposition~\ref{prop:subrepr-fun}. In categorical terms, Theorem~\ref{thm:gsetrepr} restricts to:

\begin{thm}~\label{thm:subgrprepr}
For data symmetry $(\D,G)$, the category $\GNomI$ is equivalent to a category with:
\begin{itemize}
\item as objects, open subgroups $H\leq G$,
\item as morphisms from $H$ to $K$, cosets $K\pi$ such that $\pi H\subseteq K\pi$.
\end{itemize}
\end{thm}
Here, $\GNomI$ denotes the category of single-orbit nominal sets and equivariant functions.

%%% Local Variables: 
%%% TeX-master: "main"
%%% End: 

\section{Well-behaved symmetries}\label{sec:least-supp}

Open subgroups of permutation groups are rather abstract entities, and it is not at all clear how to represent them by finite means. Much more concrete representations can be obtained under certain assumptions on the data symmetry involved, as we shall now demonstrate.

\subsection{Least supports}

An element of a nominal set always has {\em minimal} supports with respect to  inclusion, simply because it has some finite support.
As shown in Example~\ref{ex:intsig}, there may be many incomparable minimal supports (which means that there is no  {\em least} support).
Minimal supports of the same element might even have different  cardinalities, as illustrated by the following example.

\begin{exa}
For a permutation $\pi \in \permgrp(\Nat)$, let $\pi^2  \in \permgrp(\Nat\times\Nat)$ be the permutation
\[
\pi^2(n,m) = (\pi(n), \pi(m)).
\]
Let $\D = \Nat \times \Nat$ and let $G = \set{ \pi^2 :  \pi \in \permgrp(\Nat) } \leq \text{Sym}(\D)$.
Essentially, $G$ contains all permutations of $\Nat$, extended coor\-di\-na\-te-\-wise to $\Nat \times \Nat$.
Consider the set $\D$ as a nominal $G$-set, with the canonical action of $G$.
The pair $(0,1)$ has three minimal supports: the singleton $\set{(0,1)}$, the singleton $\set{(1,0)}$,
and the two-element set $\set{(0,0),(1,1)}$.
%Let $\D = \Nat \cup \Nat \times \Nat$ and let $G \leq \text{Sym}(\D)$ contain all permutations $\pi$
%satisfying the following conditions:
%$\pi(\Nat) = \Nat, \pi(\Nat \times \Nat) = \Nat \times \Nat$ and
%$\pi(n, m) = (\pi(n), \pi(m))$.
%Essentially, $G$ contains all permutations of $\Nat$, extended coor\-di\-na\-te-\-wise to $\Nat \times \Nat$.
%Consider $X = \Nat \times \Nat$ as a nominal $G$-set, with the action defined by
%$(m, n) \cdot \pi = \pi(m, n)$. The pair $(0,1)$ has two minimal supports: the singleton $\set{(0,1)}$ and the two-element set $%\set{0,1}$.
\end{exa}

The following fact follows immediately from the development of Section~\ref{sec:gset-repr}:

\begin{fact}\label{fact:admits-leastsupports}
A symmetry $(\D,G)$ admits least supports if and only if for every subgroup $H\leq G$ and for every finite $C,D\subseteq \D$, if $G_C\leq H$ and $G_D\leq H$ then $G_{C\cap D}\leq H$ (see~\eqref{eqn:GC}).
\end{fact}

We now give a convenient sufficient and necessary condition for $(\D,G)$ admitting least supports. It is easy to check that
\[
	C\subseteq D \mbox{ implies } G_C\geq G_D
\]
and, as a result, for all $C,D\subseteq\D$,
\[
	G_{C\cap D} \geq G_C + G_D,
\]
where the right-hand side denotes the subgroup of $G$ generated by the union of $G_C$ and $G_D$, i.e., the smallest subgroup of $G$ that contains $G_C$ and $G_D$. The opposite subgroup inclusion guarantees least supports for open subgroups of $G$. In fact it is not necessary to compare both sides as groups, but merely to check containment of their single orbits of $\D$, in the special case when both $C \setminus D$ and $D\setminus C$ are singleton sets.

\begin{thm}\label{thm:leastsupports}
For any symmetry $(\D,G)$, the following conditions are equivalent:
\begin{enumerate}
\item For all finite $E\subseteq \D$  and $c,d\in \D\setminus E$ such that $c\neq d$,
\[
	c\cdot G_E\subseteq c\cdot\left(G_{E\cup\{c\}}+G_{E\cup\{d\}}\right) .
\]
\item $(\D,G)$ admits least supports, i.e., if $G_C\leq H$ and $G_D\leq H$ then $G_{C\cap D}\leq H$, for any $H\leq G$ and any  finite $C,D\subseteq \D$.
\end{enumerate}
\end{thm}
\begin{proof}
(2)$\Longrightarrow$(1) is easy: take $C=E\cup\{c\}$, $D=E\cup\{d\}$ and $H=G_{E\cup\{c\}}+G_{E\cup\{d\}}$. Clearly $G_C\leq H$ and $G_D\leq H$, so by (2), $G_E\leq H$, hence $c\cdot G_E\subseteq c\cdot H$ for any $c\in\D$.

For (1)$\Longrightarrow$(2), we shall assume (1) and prove (2) by induction on the size of the (finite) set $C\cup D$. 

If $C\subseteq D$ or $D\subseteq C$, then $C\cap D=C$ or $C\cap D=D$ and the conclusion follows trivially. Otherwise, consider any $c\in C\setminus D$ and $d\in D\setminus C$; obviously $c\neq d$. Define
\[
	E = (C\cup D)\setminus\{c,d\}.
\]
We have $C\subseteq E\cup\{c\}$ and $D\subseteq E\cup\{d\}$, so
\[
	G_{E\cup\{c\}}\leq G_C\leq H \qquad \qquad 
	G_{E\cup\{d\}}\leq G_D\leq H.
\]

We shall now prove that $G_E\leq H$. To this end, consider any 
$\pi\in G_E$. By (1), there exists
a permutation 
\[\tau=\sigma_1\theta_1\sigma_2\theta_2\cdots\sigma_n\theta_n\]
 such that all $\sigma_i\in G_{E\cup\{c\}}$, $\theta_i\in G_{E\cup\{d\}}$, and $\tau(c)=\pi(c)$.
Since $G_{E\cup\{c\}}\leq H$ and $G_{D\cup\{d\}}\leq H$, all $\sigma_i, \theta_i \in H$, hence also $\tau\in H$.

On the other hand, clearly $G_{E\cup\{c\}}\leq G_E$ and $G_{E\cup\{d\}}\leq G_E$, so 
all $\sigma_i,\theta_i\in G_E$, therefore $\tau\in G_E$. As a result, $\tau\pi^{-1}\in G_E$. Since $\tau\pi^{-1}(c)=c$, we obtain $\tau\pi^{-1}\in G_{E\cup\{c\}}$, therefore $\tau\pi^{-1}\in H$. Together with $\tau\in H$ proved above, this gives $\pi\in H$. Thus we have proved $G_E\leq H$.

It is now easy to show that $G_{C\cap D}\leq H$. Indeed, $|C\cup E|=|C\cup D|-1$, so by the inductive assumption for $C$ and $E$, we have $G_{C\setminus\{c\}}\leq H$ (note that $C\setminus\{c\}=C\cap E$). Further, $|(C\setminus\{c\})\cup D|=|C\cup D|-1$, so $G_{C\cap D}\leq H$ (note that $(C\setminus\{c\})\cap D=C\cap D$).
\end{proof}

As an application:

\begin{cor}\label{cor:leastsupportseq}
The equality symmetry admits least supports.
\end{cor}
\begin{proof}
Consider any finite $E\subseteq\D$ and $c,d\not\in E$ such that $c\neq d$. Take any 
$e\in c\cdot G_E = \D\setminus E$. We need to show some $\pi\in G_{C\cup\{c\}}+G_{C\cup\{d\}}$ such that $\pi(c)=e$. 

There are two cases to consider. If $e\neq d$, put $\pi=(c\ e)\in G_{C\cup\{d\}}$. 
If $e=d$, take some fresh $d'\not\in E\cup\{c,d\}$ and put $\pi=\sigma\theta$, where 
\[
	\sigma = (c\ d') \in G_{C\cup\{d\}} \qquad \mbox{and} \qquad
	\theta = (d\ d') \in G_{C\cup\{c\}}.
\]
Then use Theorem~\ref{thm:leastsupports}.
\end{proof}

Corollary~\ref{cor:leastsupportseq} was first proved by Gabbay and Pitts~\cite[Prop.~3.4]{GP02}.

\begin{cor}\label{cor:leastsupportsord}
The total order symmetry admits least supports.
\end{cor}
\begin{proof}
Consider any finite $E\subseteq\D$ and $c,d\not\in E$ such that $c\neq d$. Let $l$ be the greatest element of $E$ smaller than $c$, and let $h$ be the smallest element of $E$ greater than $c$, assuming they both exist. (The cases where $c$ is smaller/greater than all elements of $E$ are similar). Then $c\cdot G_E$ is the open interval of rational numbers $(l,h)$.
Take any $e\in (l,h)$; without loss of generality assume that $e>c$. We need to show some $\pi\in G_{C\cup\{c\}}+G_{C\cup\{d\}}$ such that $\pi(c)=e$. 

The only interesting case is $d\in(c,e]$. In this case, take some $d'\in(c,d)$ and put $\pi=\sigma\theta$, where
\begin{itemize}
\item $\sigma$ is some monotone permutation that acts as identity on $(-\infty,l]\cup[d,+\infty)$ (so $\sigma\in G_{E\cup\{d\}}$) and such that $\sigma(c)=d'$,
\item $\theta$ is some monotone permutation that acts as identity on $(\infty,c]\cup[h,+\infty)$ (so $\theta\in G_{E\cup\{c\}}$) and such that $\theta(d')=e$.
\end{itemize}
Then use Theorem~\ref{thm:leastsupports}.
\end{proof}

\subsection{Fungibility}

In general, even if $G\leq\permgrp(\D)$ admits least supports, not every finite subset of $\D$ is the least support of some open subgroup of $G$ (see Example~\ref{ex:leastnotfun} below). We now characterize those subsets that are.

For any $C\subseteq \D$ and $G\leq\permgrp(\D)$, the restriction of $G$ to $C$ is defined by
\[
	G|_C = \{\pi|_C \mid \pi\in G,\ C\cdot\pi=C\} \leq \permgrp(C).
\]
Clearly if $H\leq G$ then $H|_C\leq G|_C$.
On the other hand, for $S\leq\permgrp(C)$, the {\em $G$-extension} of $S$ is
\[
	ext_G(S) = \{\pi\in G \mid \pi|_C\in S\} \leq G.
\]

\begin{defi}\label{def:fungible}
A finite set $C\subseteq\D$ is {\em fungible} (wrt.~$G$) if for every $c\in C$ there exists a $\pi\in G$ such that:
\begin{itemize}
\item $\pi(c)\neq c$, and
\item $\pi(c')=c'$ for all $c'\in C\setminus\{c\}$. 
\end{itemize}
We say that a data symmetry $(\D,G)$ is fungible if every finite $C\subseteq \D$ is fungible.
\end{defi}

\begin{exa}
The equality symmetry and the total order symmetry are both fungible. The integer symmetry is not fungible, as the set $\{1,2\}$ is not fungible in it: if $\pi(1)=1$ then necessarily $\pi(2)=2$, for $\pi\in G$.
\end{exa}

\begin{lem}\label{lem:fungibleleastsupp}\ 
\begin{enumerate}
\item For any open $H\leq G$, if the least support of $H$ exists then it is fungible.
\item If $(\D,G)$ admits least supports then every finite fungible $C \subseteq \D$ is the least support of $ext_G(S)$, for any $S\leq\permgrp(C)$.
\item If $(\D,G)$ is fungible then every finite $C \subseteq \D$ is the least support of $ext_G(S)$, for any $S\leq\permgrp(C)$.
\end{enumerate}
\end{lem}
\begin{proof}
For (1), it is not difficult to check that if $C$ is not fungible then $G_{C\setminus\{c\}}=G_C$ for some $c\in C$, therefore whenever $C$ supports $H$ so does $C\setminus\{c\}$.

For (2), first show that $C$ supports $ext_G(S)$; indeed, $G_C = \{\pi\in G\mid \pi|_C=e|_C\} \subseteq ext_G(S)$. In this part fungibility is not used. Since $(\D,G)$ admits least supports, if $C$ is not the least support then there must be some support properly contained in it. However, if $C$ is fungible then no $C\setminus\{c\}$ supports $ext_G(S)$; indeed, the permutation $\pi$ from Definition~\ref{def:fungible} is a witness for $G_{C\setminus\{c\}}\not\leq ext_G(S)$. Since supports of a given group are always closed under supersets, no $C'\subsetneq C$ supports $ext_G(S)$. 

Note that in (3) the existence of least supports is not assumed, so it does not follow immediately from (2). For a proof of (3), first show that $C$ supports $ext_G(S)$ as in (2) above. Then assume another support $D$ of $ext_G(S)$. We shall show that necessarily $C\subseteq D$. To this end, assume to the contrary that some $c\in C\setminus D$ exists. By the assumption on $(\D,G)$ the set $C\cup D$ is fungible, so
the permutation $\pi$ from Definition~\ref{def:fungible} is a witness for $G_{C\cup D\setminus\{c\}}\not\leq ext_G(S)$. But $G_{C\cup D\setminus\{c\}}\leq G_D$, so $G_D\not\leq ext_G(S)$, contradicting the assumption on $D$.
\end{proof}

In general, there is no implication between fungibility and the existence of least supports, as the following two examples show.

\begin{exa}\label{ex:leastnotfun}
Let $\D$ be a countably infinite set with a distinguished element $d$, and let $G$ be the group of all permutations $\pi$ of $\D$ such that $\pi(d)=d$. The symmetry $(\D,G)$ is not fungible, as the set $\{d,e\}$ is not fungible for any $e\neq d$. The fact that $(\D,G)$ admits least supports can be proved along the lines of Corollary~\ref{cor:leastsupportseq}.

Note that the set $\{d,e\}$ is not the least support of any open subgroup of $G$. Indeed, $G_{\{d,e\}}=G_{\{e\}}$, so whenever $\{d,e\}$ supports a subgroup, so does $\{e\}$.
\end{exa}

\begin{exa}\label{ex:funnotleast}
Let $\D=\{0,1\}\times\N$, and let $G$ be the group of all bijections $\pi$ on $\D$ that either preserve the first components of all elements, or negate the first components of all elements. Such a permutation may be presented by a triple $(a,\pi,\sigma)$ with $a\in \{0,1\}$ and $\pi,\sigma\in\permgrp(\N)$, acting on $\D$ as follows:
\[
	(0,n) \mapsto (a,\pi(n)) \qquad \qquad (1,n) \mapsto (1-a,\sigma(n))
\]
It is easy to check that $\D$ is fungible. Now consider the set $X=\{0,1\}$ with an action of $G$ defined by:
\[
	0 \cdot(a,\pi,\sigma) = a \qquad \qquad 1 \cdot (a,\pi,\sigma)  = 1-a
\]
Note that this action disregards the $\pi$ and $\sigma$ components of a permutation in $G$. Now, $0\in X$ is supported by any singleton $\{(0,n)\}\subseteq\D$, but not by the empty set. As a result, $(\D,G)$ does not admit least supports.
\end{exa}

\subsection{Support representation}

From now on, we assume a data symmetry $(\D,G)$ that admits least supports.

\begin{defi}\label{def:supgrprepr}
A {\em support representation} is a pair $(C,S)$, where $C\subseteq\D$ is finite and fungible, and $S\leq G|_C$. Its {\em subgroup semantics} is 
\[
\supgsem{C,S}=ext_G(S).
\]
\end{defi}
By Lemma~\ref{lem:fungibleleastsupp}(2), $\supgsem{C,S}$ is an open subgroup of $G$ and $C$ is the least support of it.
 
\begin{prop}\label{prop:supgrprepr}
Every open subgroup $H\leq G$ is equal to some $\supgsem{C,S}$.
\end{prop}
\begin{proof}
Put $S=H|_C$ where $C$ is the least support of $H$; obviously $H|_C\leq G|_C$ since $H\leq G$, and $C$ is fungible by Lemma~\ref{lem:fungibleleastsupp}(1). Then calculate
\begin{align*}
	ext_G(H|_C) &= \{\pi\in G\mid \pi|_C\in H|_C\} = \{\pi\in G\mid \exists\sigma\in H.\ \pi|_C=\sigma|_C, C\cdot\sigma=C\} \\
	&\stackrel{(\ast)}{=} \{\pi\in H\mid C\cdot\pi=C\} \stackrel{(\ast\ast)}{=} H
\end{align*}
Step $(\ast)$ above is valid since $C$ supports $H$, as $\pi|_C=\sigma|_C$ iff 
$\pi\in\mathcal{B}_C(\sigma)\subseteq H$ for $\sigma\in H$ (see~\eqref{eq:ball}). For step $(\ast\ast)$, check that
for any $\pi\in G$, 
\[
	\{\sigma^{-1}\mid \sigma\in \mathcal{B}_C(\pi)\} = \mathcal{B}_{C\cdot\pi}(\pi^{-1}).
\] 
This implies that if $C$ supports $H$ then so does $C\cdot\pi$, for any $\pi\in H$. Since $C$ is the least support of $H$, there must be $C\subseteq C\cdot\pi$ and hence by finiteness, $C\cdot\pi=C$.
\end{proof}

In the following we shall use a simple characterization of the subgroup relation in terms of representations:

\begin{lem}\label{lem:subgrp}
$\supgsem{C,S}\leq\supgsem{D,T}$ if and only if $D\subseteq C$ and $S|_D\leq T$.
\end{lem}
\begin{proof}
First we prove that $\supgsem{C,S}\leq\supgsem{D,T}$ implies $D\subseteq C$. Indeed, assuming the former, $C$ supports $\supgsem{D,T}$ (as it supports $\supgsem{C,S}$). However, the least support of $\supgsem{D,T}$ is $D$ by Lemma~\ref{lem:fungibleleastsupp}(2), therefore $D\subseteq C$.

Then, assuming $D\subseteq C$, unfold the definitions and check
\[\begin{array}{c}
   \supgsem{C,S}\leq\supgsem{D,T} \\
   \Updownarrow \\ 
   \forall \pi\in G.\ \pi|_C\in S \Longrightarrow \pi|_D\in T \\
   \Updownarrow \\
   \forall \pi\in G.\ \pi|_C\in S \Longrightarrow (\pi|_C)|_D\in T \\
   \Updownarrow \\
   \forall \tau\in S.\ \tau|_D\in T;
\end{array}\]
the last step uses the assumption that $S\leq G|_C$.
\end{proof}

We now compose representations~\ref{def:subrepr} and~\ref{def:supgrprepr} to represent single-orbit nominal $G$-sets in terms of least supports.

\begin{defi}\label{def:supsetrepr-def}
The {\em $G$-set semantics} $\supssem{C,S}$ of a support representation (see Definition~\ref{def:supgrprepr}) is the set of those functions $u:C\to\D$ that extend to a permutation from $G$, quotiented by the equivalence relation:
\begin{equation}\label{eqn:equivK}
	u\equiv_S v \iff \exists \tau\in S.\ \tau u=v.
\end{equation}
An action of $G$ on $\supssem{C,S}$  is defined by composition: 
\[
\absclass{S}{u}\cdot\pi=\absclass{S}{u\pi}.
\]
\end{defi}
Here and in the following, by $\absclass{S}{u}$ we denote the equivalence class of $u$ under $\equiv_S$.

\begin{prop}\label{prop:supsetrepr} 
(1) $\supssem{C,S}$ is a single-orbit nominal $G$-set.
(2) Every single-orbit nominal $G$-set $X$ is isomorphic to some $\supssem{C,S}$.
\end{prop}
\begin{proof}
Both parts easily follow from Propositions~\ref{prop:subrepr} and~\ref{prop:supgrprepr}
once we prove that
\begin{equation}\label{eq:reprcomp}
	\supssem{C,S}\cong\cossem{\supgsem{C,S}}.
\end{equation}
(recall from Definition~\ref{def:subrepr} that $\cossem{H}$ is the set of right cosets of a subgroup $H$ in $G$).
For this we need an equivariant bijection between $\supssem{C,S}$ and the set of cosets of $H=\supgsem{C,S}$ in $G$.

To this end, map a coset $H\sigma$ to $\absclass{S}{\sigma|_C}$; this is well-defined since $C$ supports $H$. 
Conversely, for any $u:C\to\D$, map $\absclass{S}{u}$ to $H\sigma$ where $\sigma\in G$ is such that $\sigma|_C=u$. This is again well-defined under the choice of $\sigma$ since $C$ supports $H$. To check that it is also well-defined under the choice of $u$ from $\absclass{S}{u}$, assume $\tau u = v$ for some $\tau\in S$. Since $H=ext_G(S)$, there is some $\pi\in H$ such that $\pi|_C=\tau$. Then $\sigma|_C=u$ and $\theta|_C=v$ implies $(\pi\sigma)|_C=\theta|_C$, therefore (since $C$ supports $H$) $H\sigma=H\pi\sigma=H\theta$.

Finally, it is easy to check that the two constructions are equivariant and mutually inverse.
\end{proof}

It is also possible to represent equivariant functions between $G$-sets represented via least supports.

\begin{prop}\label{prop:supfunrepr}
Let $X=\supssem{C,S}$ and $Y=\supssem{D,T}$ be single-orbit nominal sets. Equivariant functions from $X$ to $Y$ are in bijective correspondence with those injective functions $u:D\to C$ that extend to a permutation from $G$, such that $uS\subseteq Tu$, quotiented by $\equiv_T$ (see~Definition~\ref{def:supsetrepr-def}). 
\end{prop}
\begin{proof}
By Proposition~\ref{prop:subrepr-fun} and by~\eqref{eq:reprcomp}, equivariant functions from $X$ to $Y$ bijectively correspond to those elements $[u]_T\in\supssem{D,T}$ (i.e., injective functions $u:D\to\D$ that extend to permutations from $G$, quotiented by $\equiv_T$) for which the condition
\begin{equation}\label{eqn:weiljgv}
   \supgsem{C,S}\leq G_{\absclass{T}{u}}
\end{equation}
holds. Considering $\absclass{T}{u}$ as a right coset of $\supgsem{C,K}$, it is easy to show that $G_{\absclass{T}{u}}=\pi^{-1}\supgsem{D,T}\pi$, for any $\pi\in G$ that extends $u$. Further, it is easy to check that $\pi^{-1}\supgsem{D,T}\pi=\supgsem{D\cdot u,u^{-1}Tu}$ (here note that $D\cdot u$ is fungible whenever $D$ is). As a result,~\eqref{eqn:weiljgv} is equivalent to
\[
	\supgsem{C,S}\leq\supgsem{D\cdot u,u^{-1}Tu}
\]
and, by Lemma~\ref{lem:subgrp}, to
\[
	D\cdot u\subseteq C \qquad\mbox{and}\qquad S|_{D\cdot u}\leq u^{-1}Tu.
\]
Equivalently, $u$ is an injection from $D$ to $C$ such that $uS\subseteq Tu$, as in the conclusion.
\end{proof}

As before, Propositions~\ref{prop:supsetrepr} and~\ref{prop:supfunrepr} can be phrased in the language of category theory, by analogy to Theorem~\ref{thm:subgrprepr}:

\begin{thm}\label{thm:suprepr}
For any data symmetry $(\D,G)$ which admits least supports, the category $\GNomI$ is equivalent to a category with:
\begin{itemize}
\item as objects, pairs $(C,S)$ where $C\subseteq\D$ is finite and fungible and $S\leq G|_C$,  % $S\leq\permgrp(C)$,
\item as morphisms from $(C,S)$ to $(D,T)$, those injective functions $u:D\to C$ that extend to permutations from $G$, such that $uS\subseteq Tu$, quotiented by $\equiv_T$.
\end{itemize}
\end{thm}\smallskip

\noindent This representation is much more concrete than those of Theorems~\ref{thm:gsetrepr} or~\ref{thm:subgrprepr}. Indeed, pairs $(C,S)$ are finite entities, and equivariant functions are also represented by finite functions. As an immediate application, we obtain:

\begin{cor}\label{cor:countable}
For any data symmetry $(\D,G)$ with $\D$ countable, which admits least supports, there are only countably many non-isomorphic single-orbit nominal $G$-sets.
\end{cor}
\begin{proof}
Since $\D$ is countable, it has only countably many finite subsets $C$. Moreover, for any $C$, there are only finitely many choices of $S\leq\permgrp(C)$.
\end{proof}

To obtain an even more appealing representation, we shall now restrict attention to symmetries arising from certain classes of finite relational structures.

\section{\Fra symmetries}\label{sec:fraisse}

Two of the key symmetries studied in this paper: the equality and the total order symmetry, 
arise from a general construction of a \Fra limit known from standard model theory, to be defined in this section.

\subsection{\Fra limits}

A {\em signature} is a set of relation names together with (finite) arities. We shall now consider relational structures
over some fixed finite signature. For two relational structures $\struct$ and $\structB$, an \emph{embedding} $f : \struct \to \structB$ in an injective function  from the  carrier of $\struct$ to the carrier of $\structB$ that preserves and reflects all relations in the signature.

\begin{defi}\label{def:fraisseclass}
A class $\KK$ of finite structures over some fixed signature is called a \emph{\Fra class} if it:
\begin{itemize}
\item is closed under isomorphisms and substructures,
\item has \emph{amalgamation}: if $f_\structB: \structA \to \structB$ and $f_\structC: \structA \to \structC$ are embeddings and $\structA, \structB, \structC \in \KK$ then there
is a structure $\structD \in \KK$ together with two embeddings $g_\structB: \structB \to \structD$ and $g_\structC: \structC \to \structD$ that agree on the 
images of $f_\structB$ and $f_\structC$, i.e., 
$g_\structB f_\structB = g_\structC f_\structC$.
\end{itemize}
\end{defi}

\begin{exa}\label{ex:caowenvq}
Examples of \Fra classes include, over the empty signature:
\begin{enumerate}[label=\({\alph*}]
\item all finite structures, i.e.,~sets,
\item sets of size at most $k$, for any constant $k>0$,
\end{enumerate}
over the signature with a single unary predicate symbol $P$:
\begin{enumerate}[resume,label=\({\alph*}]
\item all finite sets such that at most one element satisfies $P$,
\end{enumerate}
and over the signature with a single binary relation symbol:
\begin{enumerate}[resume,label=\({\alph*}]
\item all finite structures, i.e.,~directed graphs,
\item undirected graphs, undirected trees,
\item equivalence relations,
\item equivalence relations with at most two equivalence classes,
\item preorders, partial orders, total orders.
\end{enumerate}
Classes that are {\em not} \Fra due to lack of amalgamation include, over the signature with a single binary relation symbol:
\begin{enumerate}[resume,label=\({\alph*}]
\item total orders of size at most $k$, for any constant $k>1$,
\item directed acyclic graphs,
\item undirected forests (i.e., sets of disjoint trees),
\item planar graphs.
\end{enumerate}
\end{exa}

The following theorem is standard in model theory (see e.g.~\cite{hodges}):

\begin{thm}
For any \Fra class $\KK$ there exists a unique, up to isomorphism, countable \emph{universal structure} $\U$, called the \Fra limit of $\KK$,
such that:
\begin{itemize}
\item the class of structures isomorphic to finite substructures of $\U$ is exactly $\KK$, and
\item $\U$ is \emph{homogenous}, i.e.,
 any isomorphism between two finite substructures of $\U$ 
 extends (not necessarily uniquely) to an automorphism of $\U$.
\end{itemize}
\end{thm}\smallskip

\noindent For the rest of this section, fix a \Fra class $\KK$.
From $\KK$ we obtain a data symmetry $(\DU,\GU)$, where $\DU$ is the carrier of
$\U$ and $\GU=\autgrp(\U)\leq\permgrp(\DU)$ is its group of automorphisms. We shall call a data symmetry of this form a {\em \Fra symmetry}.

\begin{exa}\label{ex:orenve}
The equality and total order symmetries (see Example~\ref{ex:symmetries}), are both \Fra symmetries; the former arises from the class of all finite sets, the latter from the class of finite total orders.

Other \Fra symmetries of interest include:
\begin{itemize}
\item The {\em graph symmetry}, arising from the class of finite undirected graphs. The universal undirected graph is the so-called random graph~\cite{rado}, where vertices are natural numbers, and an edge $\{x,y\}$ is present if and only if the $x$-th bit in the binary representation of $y$ is $1$ (for $x<y$). In the graph symmetry, $\D$ is therefore the set of natural numbers, and $G$ is the automorphism group of the random graph.
\item The {\em partial order symmetry}, arising from the class of finite partial orders. The universal structure $\U$ is not easily described in this case (see~e.g.~\cite{upo2}), except that it is partially ordered and homogenous. 
\end{itemize}
\end{exa}

\begin{defi}\label{defi:well-behaved}
A \Fra symmetry $(\DU,\GU)$ is {\em well-behaved} if
%\begin{itemize}
%\item it is a \Fra symmetry arising from a class $\KK$,
% \item $\KK$ is \emph{effective}, i.e., has decidable membership;  SL: wyrzucilem
%\item 
it admits least supports and is fungible.
%\end{itemize}
\end{defi}

All symmetries in Example~\ref{ex:orenve} are well behaved. However,
not every \Fra symmetry admits least supports or is fungible. Indeed, symmetries in Examples~\ref{ex:leastnotfun} and~\ref{ex:funnotleast} are both {\Franosp}. The one from Example~\ref{ex:leastnotfun} arises from the class in Example~\ref{ex:caowenvq}(c), and the one from Example~\ref{ex:funnotleast}  from Example~\ref{ex:caowenvq}(g).

\subsection{Structure representation}

We shall now refine the nominal set representation provided in Section~\ref{sec:least-supp} for well-behaved \Fra symmetries.
Looking at Definitions~\ref{def:supgrprepr} and~\ref{def:supsetrepr-def}, from the properties of \Fra limits it is easy to form the following definition:

\begin{defi}\label{def:repr}
A {\em structure representation} is a finite structure $\structA\in\KK$ (the {\em shape}) together with a group of automorphisms $S\leq\autgrp(\structA)$ (the {\em local symmetry}). Its {\em semantics} $\fsem{\structA,S}$ is the set of embeddings $u:\structA\to\U$, quotiented by $\equiv_S$ (see~\ref{eqn:equivK}).
A $\GU$-action on $\fsem{\structA,S}$ is defined by composition of embeddings with automorphisms of $\U$.
\end{defi}

\begin{prop}\label{prop:reprFra}
(1) $\fsem{\structA,S}$ is a single-orbit nominal $\GU$-set. (2) Every single-orbit nominal $\GU$-set $X$ is isomorphic to some $\fsem{\structA,S}$.
\end{prop}
\begin{proof}
Easy from Proposition~\ref{prop:supsetrepr}. Indeed, compare Definitions~\ref{def:repr} and~\ref{def:supgrprepr} and notice that $\autgrp{\structA}=(\GU)|_C$, where $C$ is the carrier of $\structA$, as $\U$ is homogenous. Moreover, embeddings of $\structA$ into $\U$ are exactly those injective functions from $C$ to $\DU$ that extend to automorphisms of $\U$. As a result, $\fsem{\structA,S}=\supssem{C,S}$.
\end{proof}

Equivariant functions get a similar characterization:

\begin{prop}\label{prop:funrepr}
Let $X=\fsem{\structA,S}$ and $Y=\fsem{\structB,T}$ be single-orbit nominal sets. Equivariant functions from $X$ to $Y$ are in bijective correspondence with those embeddings $u:\structB\to\structA$ for which $uS\subseteq Tu$, quotiented by $\equiv_T$. 
\end{prop}
\begin{proof}
Easy from Proposition~\ref{prop:supfunrepr}.
\end{proof}

As before, this induces an equivalence of categories:

\begin{thm}\label{thm:reprcatequiv}
In a well-behaved \Fra symmetry, the category $\GNomI$ is equivalent to a category %$\KRepr$ which has:
with:
\begin{itemize}
\item as objects, pairs $(\structA,S)$ where $\structA\in\KK$ and $S\leq\autgrp(\structA)$,
\item as morphisms from $(\structA,S)$ to $(\structB,T)$, those embeddings $u:\structB\to\structA$ for which $uS\subseteq Tu$, quotiented by $\equiv_T$.
\end{itemize}
\end{thm}\smallskip

\noindent For $\KK$ the class of all finite sets, this gives rise to the category of ``named sets with symmetries'' studied in the theory of history-dependent automata. In this special case, Theorem~\ref{thm:reprcatequiv} was proved in~\cite{gadducci-etal,staton-thesis}.

\subsection{Representation of Cartesian products}

Nominal automata as studied in Section~\ref{sec:nominalgautomata} are algebraic structures that involve equivariant functions, or relations, between nominal sets that are Cartesian products of other sets. To present such models by finite means, it is therefore necessary to calculate Cartesian products of nominal sets in terms of their representations. We shall now do this for the case of well-behaved \Fra symmetries.%
\footnote{ In the case of the equality symmetry, a somewhat less concrete representation, in terms of minimal spans of representation morphisms, was provided in~\cite{ciancia-thesis}.}

First, consider a Cartesian product of the form
\[
	\fsem{\structA,1}\times\fsem{\structB,1}
\]
for some finite structures $\structA,\structB\in\KK$, where both representation symmetries are trivial groups. Recall that $\fsem{\structA,1}$ is the set of embeddings $f:\structA\to\U$, with $\GU$-action defined by $f\cdot\pi=\pi\circ f$, and similarly for $\fsem{\structB,1}$.

For any pair of embeddings $f:\structA\to\U$, $g:\structB\to\U$, consider a relation $\rho_{(f,g)}$ between the carriers $A,B$ of $\structA,\structB$ defined by:
\begin{equation}\label{eq:vnawpa}
	\rho_{(f,g)}(a,b) \quad\iff\quad f(a)=g(b).
\end{equation}
Since both $f$ and $g$ are embeddings, $\rho_{(f,g)}$ is a partial isomorphism between $\structA$ and $\structB$. This isomorphism is invariant under the action of $\GU$ on pairs of embeddings:
\begin{equation}\label{eq:oibijvnwv}
	\rho_{(f,g)\cdot\pi} = \rho_{(f,g)}
\end{equation}
for all $\pi\in\GU$. Indeed, calculate:
\[
	\rho_{(f,g)\cdot\pi}(a,b) \iff (f\cdot\pi)a=(g\cdot\pi)a \iff \pi(f(a))=\pi(g(b)) \iff f(a)=g(b) \iff \rho_{(f,g)}(a,b).
\]

For a partial bijection $\rho$ between $A$ and $B$, the {\em amalgamated sum} $A\cup_{\rho}B$ is the disjoint union of $A$ and $B$ quotiented by $\rho$, together with canonical injections 
\begin{equation}\label{eq:voane}
\xymatrix{A\ar[r]^-{i} & A\cup_{\rho}B & B\ar[l]_-{j}.}
\end{equation}
To save space, $A\cup_{\rho_{(f,g)}}B$ will be denoted by $A\cup_{(f,g)}B$.

Define a function $\gamma_{(f,g)}:A\cup_{(f,g)}B\to \D$ by cases:
\begin{equation}\label{eq:vqonqwrf}
	\gamma_{(f,g)}(i(a)) = f(a) \qquad\qquad \gamma_{(f,g)}(j(b))=g(b).
\end{equation}
This is well defined by definition of $A\cup_{(f,g)}B$. Moreover, obviously
\begin{equation}\label{eq:wairenjb}
	\gamma_{(f,g)\cdot\pi} = \pi\circ \gamma_{(f,g)}.
\end{equation}
Let $\structC_{(f,g)}$ be the unique relational structure on the carrier $A\cup_{(f,g)}B$ that makes $\gamma_{(f,g)}$ an embedding into $\U$. By universality of $\U$, we have $\structC_{(f,g)}\in\KK$.
From~\eqref{eq:wairenjb} it is clear that
\begin{equation}\label{eq:vanojwvw}
	\structC_{(f,g)\cdot\pi} = \structC_{(f,g)}
\end{equation}
for any $\pi\in\GU$. Also, it is easy to see that $i:\structA\to\structC_{(f,g)}$ and $j:\structB\to\structC_{(f,g)}$ are embeddings.

In sum, embeddings $f:\structA\to\U$ and $g:\structB\to\U$ determine:
\begin{itemize}
\item a partial isomorphism $\rho_{(f,g)}$ between $\structA$ and $\structB$,
\item a relational structure $\structC_{(f,g)}$ on $A\cup_{(f,g)}B$,
\item an embedding $\gamma_{(f,g)}:\structC_{(f,g)}\to\U$;
\end{itemize}
moreover, by~\eqref{eq:oibijvnwv} and~\eqref{eq:vanojwvw}, $\rho_{(f,g)}$ and $\structC_{(f,g)}$ are invariant under the $\GU$-action on $(f,g)$. 
As a result, we obtain an equivariant function to a disjoint union:
\begin{equation}\label{eq:aowejnva}
	\fsem{\structA,1}\times\fsem{\structB,1} \longrightarrow \coprod_{\rho,\structC}\fsem{\structC,1}
\end{equation}
where $\rho$ ranges over partial isomorphisms between $\structA$ and $\structB$, and $\structC\in\KK$ over those relational structures on $A\cup_{\rho}B$ that make the inclusions $i,j$ in~\eqref{eq:voane} embeddings. 
In other words, $(\rho, \structC)$ ranges over the indexing set:
\begin{equation} \label{eq:indexingset}
I_{\structA, \structB} \ = \ \set{ (\rho_{(f,g)}, \structC_{(f,g)}) : f:\structA\to\U \text{ and } g:\structB\to\U}.
\end{equation}
It is not difficult to define an inverse to~\eqref{eq:aowejnva}: given $\rho$ and $\structC$, simply precompose embeddings $\gamma:\structC\to\U$ with $i:\structA\to\structC$ and $j:\structB\to\structC$. Routine calculation shows that both constructions are mutually inverse, therefore~\eqref{eq:aowejnva} is an isomorphism of nominal sets.
	
If the relational signature of $\KK$ is finite and the class $\KK$ has decidable membership, then the collection of all possible $\rho$ and $\structC$ is finite and can be effectively enumerated. As a result, we have obtained a way to compute representations of Cartesian products of the form $\fsem{\structA,1}\times\fsem{\structB,1}$.

We now adapt the above reasoning to the general case
\[
  \fsem{\structA,S}\times\fsem{\structB,T}
\]
for arbitrary $S\leq\autgrp(\structA)$ and $T\leq\autgrp(\structB)$. 

First, consider an action of the product group $S^{op}\times T$ on the set of partial isomorphisms between $\structA$ and $\structB$ defined by:
\[
	(\rho\cdot(\sigma,\tau))(a,b) \quad \iff \quad \rho(\sigma(a),\tau^{-1}(b));
\]
equivalently, with $\rho$ considered as a partial isomorphism {\em from} $\structA$ {\em to} $\structB$, this can be written as
\[
	\rho\cdot(\sigma,\tau) = \tau\circ\rho\circ \sigma.
\]
For any $\sigma\in S$ and $\tau\in T$, there is a bijection
\[
	m_{\sigma,\tau} : A\cup_{\rho}B \to A\cup_{\rho\cdot(\sigma,\tau)}B
\]
given by:
\begin{equation}\label{eq:dfbtnrn}
	m_{\sigma,\tau}(i(a)) = i(\sigma^{-1}(a)) \qquad \qquad 
	m_{\sigma,\tau}(j(b)) = j(\tau(b)).
\end{equation}
This is well defined; indeed, calculate:
\[
	i(a)=j(b) \iff \rho(a,b) \iff (\rho\cdot(\sigma,\tau))(\sigma^{-1}(a),\tau(b)) \iff i(\sigma^{-1}(a))=j(\tau(b)).
\]
For any relational structure $\structC$ on $A\cup_{\rho}B$, let $\structC\cdot(\sigma,\tau)$ be the unique structure on $ A\cup_{\rho\cdot(\sigma,\tau)}B$ that makes $m_{\sigma,\tau}$ into an isomorphism. 

It is easy to check that we thus obtain a group action of $S^{op}\times T$ on the indexing set~\eqref{eq:indexingset} of the disjoint union in~\eqref{eq:aowejnva}.
% i.e., the set of pairs $(\rho,\structC)$ where $\rho$ is a partial isomorphism between $\structA$ and $\structB$, and $\structC\in\KK$ is a relational structure on $A\cup_{\rho}B$ such that $i:\structA\to\structC$ and $j:\structB\to\structC$ are embeddings. 
Pick a family of representatives $(\underline{\rho},\underline{\structC})$ for each orbit of this action.
For any representative, where $\underline{\structC}\in\KK$ is a structure on $A\cup_{\underline{\rho}}B$, let $S\Cup T\leq\autgrp(\underline{\structC})$ be the group of all those automorphisms of $\underline{\structC}$ that, roughly speaking, restrict to $S$ on $\structA$ and to $T$ on $\structB$.
Formally, 
\begin{equation}\label{eq:nwirbn}
S\Cup T \ \ = \ \ i^{-1} S i \ \cap  \ j^{-1} T j.
\end{equation}

The following theorem is a generalization of~\eqref{eq:aowejnva}.

\begin{thm}\label{thm:cartesian}
There is an equivariant isomorphism
\[
	\fsem{\structA,S}\times\fsem{\structB,T} \quad \cong \quad \coprod_{\underline\rho,\underline\structC}\fsem{\underline\structC,S\Cup T}
\]
where $\underline\rho,\underline\structC$ in the disjoint union range over the chosen representatives as above.
\end{thm}
\begin{proof}
For a function from left to right, take any $\absclass{S}{f}\in\fsem{\structA,S}$ and $\absclass{T}{g}\in\fsem{\structA,T}$. The embeddings $f:\structA\to\U$ and $g:\structB\to\U$ determine a partial isomorphism $\rho_{(f,g)}$, a relational structure $\structC_{(f,g)}$ and an embedding $\gamma_{(f,g)}:\structC_{(f,g)}\to\U$ as before.

Let $(\underline\rho,\underline\structC)$ be the chosen representative of the $(S^{op}\times T)$-orbit of $(\rho_{(f,g)},\structC_{(f,g)})$. In particular, there exists some $\sigma\in S$ and $\tau\in T$ such that
\begin{equation}\label{eq:nvbrae}
	\rho_{(f,g)} =\underline\rho\cdot(\sigma,\tau) \qquad\qquad \structC_{(f,g)} = \underline\structC\cdot(\sigma,\tau).
\end{equation}

Define an embedding $\gamma:\underline\structC\to\U$ by:
\begin{equation}\label{eq:mymftg}
	\gamma = \gamma_{(f,g)}\circ m_{(\sigma,\tau)}.
\end{equation}

There may be many possible choices of $\sigma,\tau$ that
satisfy~\eqref{eq:nvbrae}, and they may yield different embeddings $\gamma$. However, all these embeddings are $\equiv_{S\Cup T}$-equivalent. To see this, assume
\begin{equation*}
	\underline\rho\cdot(\sigma,\tau) =\underline\rho\cdot(\sigma',\tau') \qquad\qquad 
	\underline\structC\cdot(\sigma,\tau) = \underline\structC\cdot(\sigma',\tau');
\end{equation*}
then it is easy to check 
\[
	m_{(\sigma,\tau)} = m_{(\sigma',\tau')}\circ m_{(\sigma'^{-1}\sigma,\tau\tau'^{-1})},
\]
and $m_{(\sigma'^{-1}\sigma,\tau\tau'^{-1})}$ is an automorphism of $\underline\structC$ that restricts to $\sigma'^{-1}\sigma\in S$ on $\structA$ and to $\tau\tau'^{-1}\in T$ on $\structB$. As a result,
\[
	m_{(\sigma,\tau)}\equiv_{S\Cup T}m_{(\sigma',\tau')}.
\]
Moreover, for $\gamma$ in~\eqref{eq:mymftg}, $[\gamma]_{S\Cup T}$  does not depend on the choice of representatives $f\in[f]_S$ and $g\in [g]_T$. To see this, notice that for any $\sigma\in S$ and $\tau\in T$:
\[
	\gamma_{(f\circ\sigma,g\circ\tau)} = \gamma_{(f,g)}\circ m_{(\sigma^{-1},\tau)}
\]
by~\eqref{eq:vqonqwrf} and~\eqref{eq:dfbtnrn}.

As a result, we obtain a function that maps the pair $([f]_S,[g]_T)$ to $\underline{\rho}$, $\underline{\structC}$ and $[\gamma]_{S\Cup T}$. Equivariance of this function is checked routinely. As before, its inverse is obtained by precomposing embeddings $h:\underline\structC\to\U$ with injections $i:\structA\to\structC$ and $j:\structB\to\structB$. Both constructions are well-defined and mutually inverse up to $\equiv_S$, $\equiv_T$ and $\equiv_{S\Cup T}$.
\end{proof}

\begin{exa}\label{ex:voenvw}
In the equality symmetry, where $\KK$ is the class of finite sets, there are no nontrivial relational structures, i.e., every structure is simply its carrier. Let
\[
	\structA = \{x,y\} \qquad\qquad \structB = \{z\}.
\]
By Definition~\ref{def:repr}, there is
\[
	\fsem{\structA,1} \cong \D^{(2)} \qquad\qquad \fsem{\structB,1} \cong \D
\]
(see Example~\ref{ex:gsets}).
There are three partial isomorphisms between $\structA$ and $\structB$:
\[
	\rho_1 = \{(x,z)\} \qquad\qquad \rho_2 = \{(y,z)\} \qquad\qquad \rho_3 = \emptyset,
\]
with the corresponding amalgamated sums $A\cup_{\rho_i}B$ having $2$, $2$ and $3$ elements, respectively. By~\eqref{eq:aowejnva}, there is an isomorphism
\[
	\D^{(2)}\times \D \cong \D^{(2)} + \D^{(2)} + \D^{(3)}
\]
(here and in the following, $+$ denotes disjoint union).
In elementary terms:
\[
	\{(c,d) \mid c\neq d\} \times \D = \{(c,d,e)\mid c\neq d=e\} + \{(c,d,e)\mid e=c\neq d\} + \{(c,d,e)\mid c\neq d\neq e\}.
\]
In general, the product $\D^{(n)}\times \D$ has $n+1$ orbits.

Now consider a local symmetry $S = \autgrp(\structA) = \{1,(x\ y)\}$. By Definition~\ref{def:repr}, there is
\[
	\fsem{\structA,S} \cong \choose{\D}{2}
\]  
(see Example~\ref{ex:gsets}). Partial isomorphisms $\rho_1$ and $\rho_2$ form an orbit under the action of $S^{op}\times 1$, therefore by Theorem~\ref{thm:cartesian}, the product of $\fsem{A,S}$ and $\fsem{B,1}$ has only two orbits:
\[
	\choose{\D}{2} \times \D \cong \fsem{\{x,y\},1} + \fsem{\{x,y,z\},S\Cup 1} \cong \D^{(2)} + \big\{(\{x,y\},z)\mid z\not\in\{ x,y\} \big\};
\]
here $S\Cup 1 = \{1,(x\ y)\}$.
\end{exa}

\begin{exa}\label{ex:voenvw2}
In the total order symmetry, where $\KK$ is the class of finite total orders, there are no nontrivial local symmetries $S$ in representations, since the only automorphism of a finite total order is the identity. Let
\[
	\structA = \{x<y\} \qquad\qquad \structB = \{z\}.
\]
By Definition~\ref{def:repr}, there is
\[
	\fsem{\structA,1} \cong \D^{(<2)} \qquad\qquad \fsem{\structB,1} \cong \D
\]
(see Example~\ref{ex:gsets}).
As in Example~\ref{ex:voenvw}, there are three partial isomorphisms between $\structA$ and $\structB$:
\[
	\rho_1 = \{(x,z)\} \qquad\qquad \rho_2 = \{(y,z)\} \qquad\qquad \rho_3 = \emptyset.
\]
However, in this case there are three different total orders on $A\cup_{\rho_3}B=\{x,y,z\}$ that embed $\structA$ and $\structB$:
\[
	\structC = \{z<x<y\}, \qquad\qquad \structC' = \{x<z<y\}, \qquad\qquad \structC'' = \{x<y<z\},
\]
each giving rise to a different orbit of the Cartesian product. As a result, there are five orbits:
\[
	\D^{(<2)}\times\D = \D^{(<2)} + \D^{(<2)} + \D^{(<3)} + \D^{(<3)} + \D^{(<3)}.
\]
In general, the product $\D^{(<n)}\times \D$ has $2n+1$ orbits.
\end{exa}

\begin{exa}\label{ex:voenvw3}
Consider the graph symmetry, where $\KK$ be the class of all finite undirected graphs (see Example~\ref{ex:orenve}).
By analogy to Example~\ref{ex:voenvw}, let
\[\begin{array}{ccc}
\xymatrix{x & y} &\qquad\qquad& \xymatrix{z} 
\\ \phantom{a} \\
\structA && \structB
\end{array}\]
be discrete graphs. There are three partial isomorphisms between $\structA$ and $\structB$:
\[
	\rho_1 = \{(x,z)\} \qquad\qquad \rho_2 = \{(y,z)\} \qquad\qquad \rho_3 = \emptyset,
\]
with the corresponding amalgamated sums $A\cup_{\rho_i}B$ having $2$, $2$ and $3$ elements, respectively. The sums corresponding to $\rho_1$ and $\rho_2$ have unique (discrete, isomorphic to $\structA$) graphs on them that embed $\structA$. The sum $A\cup_{\rho_3}B=\{x,y,z\}$ allows four graphs:
\[\begin{array}{ccccccc}
\xymatrix{z & y \\ x} &\qquad& \xymatrix{z\ar@{-}[d] & y \\ x} &\qquad& \xymatrix{z\ar@{-}[r] & y \\ x} &\qquad& \xymatrix{z\ar@{-}[d]\ar@{-}[r] & y \\ x}
\\ \phantom{a} \\
\structC && \structC' && \structC'' && \structC'''
\end{array}\]
As a result, the indexing set in~\eqref{eq:indexingset} has six elements altogether:
\begin{equation}\label{eq:voefve}
I_{\structA, \structB} = \{(\rho_1,\structA),(\rho_2,\structA),(\rho_3,\structC),(\rho_3,\structC'),(\rho_3,\structC''),(\rho_3,\structC''')\}
\end{equation}
hence the product $\fsem{\structA,1}\times\fsem{\structB,1}$ has six orbits:
\[
\fsem{\structA,1}\times\fsem{\structB,1} = \fsem{\structA,1} + \fsem{\structA,1} + \fsem{\structC,1} 
+ \fsem{\structC',1} + \fsem{\structC'',1} + \fsem{\structC''',1}.
\]

Now consider a local symmetry $S = \autgrp(\structA) = \{1,(x\ y)\}$. Under the action of the group $S^{op}\times 1$ on the three-element set of partial isomorphisms between $\structA$ and $\structB$, $\rho_1$ and $\rho_2$ fall in one orbit, and $\rho_3$ forms another by itself. As described above, this further determines an action of $S^{op}\times 1$ on the indexing set~\eqref{eq:voefve}. Here, $(\rho_1,\structA)$ and $(\rho_2,\structA)$ fall in one orbit, $(\rho_3,\structC')$ and $(\rho_3,\structC'')$ in another, and the remaining two elements form two singleton orbits. As the indexing set of representatives from Theorem~\ref{thm:cartesian} we may take:
\[
\{(\rho_1,\structA),(\rho_3,\structC),(\rho_3,\structC'),(\rho_3,\structC''')\}
\]
Easy calculation shows that the amalgamated groups $S\Cup 1$ from~\eqref{eq:nwirbn} are nontrivial on $\structC$ and $\structC'''$: in both cases, $S\Cup 1 = \{1,(x\ y)\}$.
As a result, by Theorem~\ref{thm:cartesian}, the product of $\fsem{\structA,S}$ and $\fsem{\structB,1}$ has the following representation:
\[
\fsem{\structA,S}\times\fsem{\structB,1} = \fsem{\structA,1} + \fsem{\structC,S\Cup 1} 
+ \fsem{\structC',1} + \fsem{\structC''',S\Cup 1}.
\]	
\end{exa}

\section{\Fra automata} \label{sec:Fraisseautomata}

A deterministic orbit-finite nominal $G$-automaton,   % A \GDFA, 
understood as in Section~\ref{sec:nominalgautomata},
%~\eqref{dgm:automaton}, 
is a simple combination of a few orbit-finite nominal $G$-sets and equivariant functions between them, involving a simple Cartesian product. It is therefore natural that an effective representation of nominal sets, equivariant functions and Cartesian products extends to a similar representation of automata for \Fra symmetries.

The resulting notion is rather complex, and for mathematical reasoning about nominal automata, the more abstract definitions introduced in Sections~\ref{sec:gautomata} and~\ref{sec:nominalgautomata} seem more suitable. Even when finite representations are important, for example for the implementation of algorithms that manipulate automata, it seems more productive to formulate those algorithms in abstract terms, and then have a general-purpose programming language that can translate them to effective procedures, constructing finite representations of complex data structures (e.g., automata) implicitly. This approach was used in~\cite{BBKL12}, where nontrivial algorithms on nominal automata were formalized and implemented without spelling out an explicit finite representation of those automata.

Nevertheless, we wish to sketch a finite representation of nominal $G$-automata for two reasons. First, the representation resembles and generalizes Kaminski-Francez finite memory automata, which shows that the equivalence results of Section~\ref{sec:register-models} are not accidental; on the contrary, that the basic ingredients of finite memory automata, such as registers, appear naturally from our representation results applied to the abstract notion of a nominal automaton. Secondly, the concrete definition of \Fra{} automaton below, although complex, does not rely on notions such as support or orbit, and so they may conceivably appeal to those who wish to study automata on infinite alphabets without learning nominal sets. The concrete notion is also more in the spirit of~\cite{FK94} and~\cite{P99}, making the relation to some previous work more apparent.

Fix for the rest of this section a  class $\KK$ of structures that induces a well-behaved \Fra symmetry $(\DU,\GU)$. 
Our goal is to apply Theorems~\ref{thm:reprcatequiv} and~\ref{thm:cartesian} to 
develop a syntax (understood as a finite representation) for \GUDFA. 
At the risk of repeating some material from Section~\ref{sec:fraisse}, we unravel below the definition of 
a deterministic orbit finite nominal $\GU$-automaton. 
% obtain representation of deterministic orbit-finite $G$-automata, in an arbitrary well-behaved \Fra symmetry $(G, \D)$.

 For the sake of presentation, we only study deterministic automata, and restrict to the alphabet $\DU$.
 The general case, when the alphabet is an arbitrary orbit finite nominal $\GU$-set such as $(\DU)^2$ or $\DU \uplus \DU$, 
 may be dealt with in essentially the same way.

 The basic intuition is that
the class $\KK$ describes all possible ``memory shapes" of an automaton.

A  \Fra \emph{$\KK$-automaton} has a finite set $Q$ of control states. Each  state $q \in Q$ comes with a structure representation $(\structA_q, S_q)$.
The set of configurations in state $q$ is the nominal set $\sem{\structA_q, S_q}$.
We shall call elements of $\structA_q$ {\em registers} of state $q$, and the group $S_q$ is the {\em register symmetry}.
The  set of all configurations of an automaton is a disjoint union:
\begin{equation}\label{eq:configurations-U}
X =  \coprod_{q \in Q}	\sem{\structA_q, S_q} .  
\end{equation}
A configuration consists of a state $q \in Q$, together with a valuation $\structA_q \to \U$ 
that maps registers to data values, and preserves and reflects the structure of $\structA_q$,
with the proviso that valuations are considered equal if they differ only by a register  symmetry. 

The automaton has a set of accepting states, and an initial state. 
The structure of registers $\structA_q$ in the initial state must  be empty.

The last ingredient of the \fra automaton is a \emph{symbolic transition function} $s = \set{s_q}_{q \in Q}$
that is  used to represent an equivariant  transition function
\begin{equation}\label{eq:Fratransfunc}
\delta_s : X \times \DU \to X .
\end{equation}
The symbolic transition function is a representation of $\delta_s$ along the lines of Theorem~\ref{thm:cartesian} and Proposition~\ref{prop:funrepr}.
We define symbolic transition functions in terms of \emph{annotations}, which are simply an elementary view on the orbits of the product $X \times \DU$ as explained in Theorem~\ref{thm:cartesian}. An \emph{annotation} of 
a representation $(\structA, S)$ is a structure of one of two kinds:
either a conservative extension $\structA^* \in \KK$ of $\structA$ by one element, denoted $*$;
or the structure $\structA$ itself with additionally one distinguished element, that we denote by $*$ as well.
In either case, we identify two annotations if they are related by an automorphism $\sigma\in S$ such that $\sigma(*)=*$.
An annotation comes thus with its local symmetry, that is isomorphic either to the group $S$ itself, or to its subgroup
determined by the requirement $\sigma(*) = *$.
There are finitely many possible annotations for every $\structA$, as the relational signature is assumed to be finite.

Intuitively speaking, annotations describe the ways in which the newly read input data value ($*$) may compare to the data values in the  registers. In other words, annotations formalize the tests an automaton on the input letters.

Note that an annotation of a structure $\structA$ uniquely determines:
\begin{itemize}
\item a partial isomorphism $\rho$ between $\structA$ and a one-element structure $*$ ($\rho$ is empty if the annotation extends $\structA$ with $*$, otherwise it identifies $*$ with the distinguished element of $\structA$),
\item a relational structure on the amalgamated sum $A\cup_{\rho}\{*\}$.
\end{itemize}
In other words, by Theorem~\ref{thm:cartesian}, annotations of $\structA_q$ correspond to orbits of the Cartesian product $\fsem{\structA_q,S_q}\times\DU$. (See also Examples~\ref{ex:voenvw}-\ref{ex:voenvw3}.)

The domain of $s_q$ contains all possible annotations of $(\structA_q, S_q)$. % denote this finite set by $\KK^*_q$. 
For any annotation $\structA^*$, the value $s_q(\structA^*)$ is a state $p\in Q$ together with an embedding
 \begin{equation}\label{eq:symbtransfunc}
s_q(\structA^*) : \structA_p \to \structA^*
\end{equation}
that commutes with the local symmetries as prescribed by Proposition~\ref{prop:funrepr}.

To sum up:

\begin{defi}\label{def:Fraisseautomaton}
A \fra $\KK$-automaton consists of:
\begin{itemize}
	\item a finite set of control states $Q$;
	\item for each state $q \in Q$, a structure representation $(\structA_q, S_q)$ (see Definition~\ref{def:repr});
	\item an initial state $q_I \in Q$ with $\structA_{q_I}$ the empty structure;
	\item a set of accepting states $F \subseteq Q$;
	\item a symbolic transition function $s = \set{s_q}_{q \in Q}$ as above.
\end{itemize}
Elements of $\structA_q$ are called {\em registers} of $q$.
\end{defi}
These ingredients naturally induce a $\GU$-automaton, with a transition function~\eqref{eq:Fratransfunc} defined  as follows.
Suppose that the state in the current  configuration is  $q \in Q$ and the valuation is represented, up to register symmetry, by  $\eta : \structA_q \to \DU$. The automaton reads an  input letter  $d \in \DU$. 
Let $\eta^*$ extend $\eta$ by mapping $*$ to $d$, thus $\eta^*: \structA^* \to \DU$ is an embedding, for some annotation $\structA^* \in \KK$.
Apply $s_q$ to $\structA^*$, yielding $p \in Q$ and a function~\eqref{eq:symbtransfunc}.
The new state is $p$, and the new valuation is obtained by composing $s_q(\structA^*)$ with the extended valuation $\eta^*$, that is $\eta^* \circ s_q(\structA^*): \structA_p \to \DU$. The new valuation is an embedding, as a composition of embeddings,
and its equivalence class depends only on the equivalence class of $\eta$, thanks to the assumption that $s_q$ commutes with local symmetries.

By Theorem~\ref{thm:cartesian} and Proposition~\ref{prop:funrepr} one obtains:

\begin{thm}\label{thm:Fra-char}
For a well-behaved \Fra symmetry induced by a class $\KK$, every reachable orbit finite deterministic 
nominal $\GU$-automaton over the input alphabet $\DU$ is isomorphic to a \fra $\KK$-automaton.
\end{thm}
%A comment is in order here. The representation results from Section~\ref{sec:fraisse} apply to nominal $\GU$-sets
%only, while the notion of $\GU$-automaton lives in $\GU$-sets and does not require finite supports.
%However, Theorem~\ref{thm:Fra-char} still holds by Proposition~\ref{prop:reachablenominal}:
%whenever the alphabet is nominal, any reachable automaton also is.   
% SL: wyrzucilem

By Theorems~\ref{thm:Fra-char} and~\ref{thm:MNrev} one directly obtains:
\begin{cor}\label{cor:Fra-char}
%When $\KK$ is well-behaved
For a well-behaved \Fra symmetry induced by a class $\KK$, 
the following conditions are equivalent for a $\GU$-language $L \subseteq {\DU}^*$:
\begin{enumerate}
\item[(1)] $L$ is recognized by a \GUDFA   % orbit finite $\GU$-automaton
\item[(2)] $L$ is recognized by a % the syntactic automaton of $L$ is isomorphic to a 
\Fra $\KK$-automaton
\item[(3)] the syntactic quotient ${\DU}^*/{\equiv_L}$ is orbit finite.
\end{enumerate}
\end{cor}

\begin{exa}
For the equality symmetry, \Fra $\KK$-automata are very similar to finite memory automata studied in Section~\ref{sec:register-models}, with two differences:
\begin{itemize}
\item the number of registers varies from state to state (thus no need for undefined register values), 
\item symmetries are imposed on registers.
\end{itemize}
An even more similar model is that of history-dependent automata~\cite{P99}, where symmetries on local names were first introduced. For the equality symmetry, our \Fra $\KK$-automata are essentially a deterministic version of history-dependent automata. A connection of the latter with
finite memory automata has been tentatively made in~\cite{ciancia-tuosto}.

For the total order symmetry, a $\KK$-automaton has a totally ordered set of registers in each state, 
and valuations are monotonic.
These automata are capable of comparing data values with respect to data ordering.
It is easy to verify that \Fra $\KK$-automata (and hence also \GUDFA, by Thm~\ref{thm:Fra-char}) in this case are expressively equivalent to deterministic finite memory au\-to\-ma\-ta of~\cite{BLP10,FHL10} over totally ordered data, in the special case of a singleton alphabet.

For the graph symmetry, a $\KK$-automaton keeps a graph of registers in each state, and valuations are graph embeddings into the random graph. An automaton can test a newly read letter for edge connections with nodes stored in current registers. To our best knowledge, this kind of automaton has not been studied in the literature.
\end{exa}

\removed{
In this section, we study a totally ordered data domain. 
Formally, we fix the data values as the set of rational numbers, $\D = \Q$,
and consider the group $\Gmon$ of all bijections of $\Q$ that preserve
the natural ordering. The group contains, in particular, all linear functions
$x \mapsto a \cdot x + b$ for $a > 0$.
As far as the automaton model is concerned, nothing changes if we replace reals $\Q$ by, e.g., an open interval $(0,1)$,
real numbers or any other dense total order.

We now define a concrete automaton model that captures the notion of $\Gmon$-automata.

\paragrafik{Total-order register automata} An \emph{total-order register automaton} is like a register automaton, with an additional total order on the registers. That is, every set of register names comes with such an order.  Configurations are as in a register automaton, with the difference that valuations are required to be strictly monotone, which we write as 
\begin{equation}\label{eq:configurations-mono}
X =  \coprod_{q \in Q}	R_q \stackrel{mono} \to \D.
\end{equation}
As in the case of register automata, the transition function can be defined by a symbolic one. The new configuration depends on the old configuration, and the relative position, in the total order,  of the input letter as compared to the data values in the registers of the previous configuration. 
In particular, if $R_q = \set{r_1 \ldots r_k}$ then the domain of $s_q$ consists of the following $2k+1$ ``symbolic regions'':
\[
(-\infty, r_1), \ \set{r_1}, \ (r_1, r_2), \ \ldots, \ (r_{k-1}, r_k), \ \set{r_k}, \ (r_k,\infty) .
\]

There is an analogous result to Lemma~\ref{lem:symbolic-represent}, which says that every equivariant function can be represented in a  symbolic way.

In the case of $\Gmon$, we do not have any notion of register symmetry. This is not needed to get the following result:
\begin{theorem}\label{thm:permgrp-char}
Every reachable orbit finite $\Gmon$-automaton over input alphabet $\Q$ is isomorphic to a total-order register automaton.
\end{theorem}

The proof of this theorem is given in later sections.

\paragrafik{Variants}
An interesting variant arises if we restrict to only those monotonic permutations of reals
that preserve, say, $0$ (or, preserve a finite set of chosen values).
The induced automaton model can  test if the input number is positive, zero or negative. The automaton does not need to store $0$ explicitly  in its registers. In each state it is fixed which registers store negative (positive) values.

\subsection{Partial order automorphisms}
\label{sec:po}

Now we investigate a partially-ordered data domain.
Instead of considering a particular fixed order, we prefer to assume a \emph{universal} one.
It is well known (see for instance~\cite{F53}) that there exists a unique, up to isomorphism, countable partial order 
that contains any finite partial order as an (induced) suborder. We call it universal and denote by $(\D, \leq)$  in this section.
The group relevant for this case contains all bijections of $\D$ that are order-automorphisms, i.e.,
that preserve and reflect the order. Denote this group by $\Gpo$.

$\D$ is \emph{homogenous} in the following sense: if $C, D \subseteq \D$ are
two finite subsets and $f: C \to D$ is an isomorphism of suborders induced by $C$ and $D$,
then $f$ extends (not necessarily uniquely) to an order-automorphism of $\D$.

\begin{remark}
The cases of unordered and totally ordered data, investigated in Sections~\ref{sec:all-perms} and~\ref{sec:mono},
are in fact special cases, in the sense that the respective data domains are universal 
and homogenous, with the proviso that a class of finite discrete orders (i.e., plain finite sets) or the class of finite total orders, respectively,
is considered in place of the class of finite partial orders.
Consequently, the automaton model described below includes symmetry register automata and ordered register automata
as subclasses.

\end{remark}

\paragrafik{Partial-order register automata} A \emph{partial-order register automaton} is like a symmetry register automaton,
but additionally equipped with a partial order $\preceq_q$ on the set $R_q$ of registers in each state $q$.
Furthermore, the symmetry $S_q$ is required to contain only automorphisms of $(R_q, \preceq_q)$:
\[
S_q \leq \text{Aut}(R_q, \preceq_q)
\]
(cf.~(\ref{eq:sym}) in Section~\ref{sec:all-perms}).

The domain of a symbolic transition function $s_q$ is larger than in~(\ref{eq:transfunc}) as it must take into account
all possible order-tests of a newly read input letter against registers.
For convenience, we will use the following notation: an \emph{annotation} of partial order $(X, \leq)$ is either this order
with one distinguished element, denoted $*$, or any extension of $\leq$ to $X \cup \{*\}$.
Formally, the domain of $s_q$ contains all annotations of $(R_q, \preceq_q)$.
The function $s_q(R)$ is then any order-embedding
 \[
s_q(R) : R_p \to R, \text{ \ for some state } p,
\]
i.e., a function that preserves and reflects the order.

Now we explain how the symbolic transition function respects the symmetry $S$.
First observe that the automorphism group $S_q \leq \text{Aut}(R_q, \preceq_q)$ naturally extends
to any annotation of $(R_q, \preceq_q)$: if $* \in R_q$ then take only those automorphisms from $s_q$
that preserve $*$; otherwise extend all automorphisms from $s_q$ to preserve $*$.
\slcomm{tu mi pachnie jakimis minorami}
Denote this just described automorphism group by $S^*_q$.
The automaton is well defined if whenever the domain of $s_q(R)$ is $R_p$ then
each automorphism from $S^*_q$, restricted to the image set  $s_q(R)(R_p) \subseteq R$, is an image,
via $s_q(R)$, of some automorphism from $R_p$.

As in previous cases, we argue that the definition of a (well-defined) automaton induces naturally a $\Gpo$-automaton.
A valuation of registers $v : R_q \to \D$ in a state $q$ must be an order-embedding. 
We write:
\begin{equation}\label{eq:configurations-po}
X =  \coprod_{q \in Q}	R_q \stackrel{emb} \to \D.
\end{equation}
Note that equations~(\ref{eq:configurations}) and~(\ref{eq:configurations-mono}) are just the same 
as~(\ref{eq:configurations-po}), when all $R_q$ and $\D$ are a discrete or total order, respectively.
A configuration consists of a state $q$ and an equivalence class of valuations, similarly as in case of
symmetry register automata.
The transition function $\delta_s$ is induced analogously, due to well-definedness of the automaton.

\begin{theorem}\label{thm:permgrp-char}
Every reachable orbit finite $\Gpo$-automaton over input alphabet $\D$ is isomorphic to a partial-order register automaton.
\end{theorem}

} % \removed

%%% Local Variables: 
%%% TeX-master: "main"
%%% End: 

\bibliographystyle{plain}
\bibliography{bib}

\begin{thebibliography}{10}

\bibitem{AR94}
J.~Ad{\'a}mek and J.~Rosick{\'y}.
\newblock {\em Locally Presentable and Accessible Categories}.
\newblock Cambridge Univ. Press, 1994.

\bibitem{Adamek:1990:AAC:575450}
J.~Adamek and V.~Trnkova.
\newblock {\em Automata and Algebras in Categories}.
\newblock Kluwer Academic Publishers, 1990.

\bibitem{BLP10}
M.~Benedikt, C.~Ley, and G.~Puppis.
\newblock What you must remember when processing data words.
\newblock In {\em AMW}, volume 619 of {\em CEUR Workshop Proceedings}, 2010.

\bibitem{DBLP:journals/tcs/BjorklundS10}
H.~Bj{\"o}rklund and T.~Schwentick.
\newblock On notions of regularity for data languages.
\newblock {\em TCS}, 411(4-5):702--715, 2010.

\bibitem{datamonoids}
M.~Boja{\'n}czyk.
\newblock Data monoids.
\newblock In {\em STACS}, volume~9 of {\em LIPIcs}, 2011.

\bibitem{DBLP:conf/lics/BojanczykMSSD06}
M.~Boja{\'n}czyk, A.~Muscholl, T.~Schwentick, L.~Segoufin, and C.~David.
\newblock Two-variable logic on words with data.
\newblock In {\em LICS}, pages 7--16, 2006.

\bibitem{BBKL12}
Miko{\l}aj Boja{\'n}czyk, Laurent Braud, Bartek Klin, and S{\l}awomir Lasota.
\newblock Towards nominal computation.
\newblock In {\em Proc. POPL'12}, pages 401--412, 2012.

\bibitem{BKL11}
Miko{\l}aj Boja{\'n}czyk, Bartek Klin, and S{\l}awomir Lasota.
\newblock Automata with group actions.
\newblock In {\em Proc. LICS'11}, pages 355--364, 2011.

\bibitem{BKLT13}
Miko{\l}aj Boja{\'n}czyk, Bartek Klin, S{\l}awomir Lasota, and Szymon
  Toru{\'n}czyk.
\newblock Turing machines with atoms.
\newblock In {\em Proc. LICS'13}, 2013.

\bibitem{ciancia-thesis}
V.~Ciancia.
\newblock {\em Accessible functors and final coalgebras for named sets}.
\newblock PhD thesis, University of Pisa, 2008.

\bibitem{ciancia-tuosto}
V.~Ciancia and E.~Tuosto.
\newblock {A novel class of automata for languages on infinite alphabets}.
\newblock Technical Report CS-09-003, University of Leicester, 2009.

\bibitem{DL09}
St{\'e}phane Demri and Ranko Lazic.
\newblock {LTL} with the freeze quantifier and register automata.
\newblock {\em ACM Trans. Comput. Log.}, 10(3), 2009.

\bibitem{FHL10}
D.~Figueira, P.~Hofman, and S.~Lasota.
\newblock Relating timed and register automata.
\newblock In {\em Proc. EXPRESS'10}, volume~41 of {\em Electronic Proceedings
  in Theoretical Computer Science}, pages 61--75, 2010.

\bibitem{FK94}
N.~Francez and M.~Kaminski.
\newblock Finite-memory automata.
\newblock {\em TCS}, 134(2):329--363, 1994.

\bibitem{FK03}
N.~Francez and M.~Kaminski.
\newblock An algebraic characterization of deterministic regular languages over
  infinite alphabets.
\newblock {\em TCS}, 306(1-3):155--175, 2003.

\bibitem{GP02}
M.~Gabbay and A.~M. Pitts.
\newblock A new approach to abstract syntax with variable binding.
\newblock {\em Formal Asp. Comput.}, 13(3-5):341--363, 2002.

\bibitem{gadducci-etal}
F.~Gadducci, M.~Miculan, and U.~Montanari.
\newblock About permutation algebras, (pre)sheaves and named sets.
\newblock {\em Higher-Order and Symbolic Computation}, 19(2-3):283--304, 2006.

\bibitem{hodges}
W.~Hodges.
\newblock {\em A shorter model theory}.
\newblock Cambridge Univ. Press, 1997.

\bibitem{upo2}
J.~Hubi\v{c}ka and J.~Ne\v{s}et\v{r}il.
\newblock Universal partial order represented by means of oriented trees and
  other simple graphs.
\newblock {\em Eur. J. Comb.}, 26:765--778, 2005.

\bibitem{maclanemoerdijk}
S.~{Mac Lane} and I.~Moerdijk.
\newblock {\em Sheaves in geometry and logic: a first introduction to topos
  theory}.
\newblock Springer, 1992.

\bibitem{MP05}
U.~Montanari and M.~Pistore.
\newblock History-dependent automata: An introduction.
\newblock In {\em SFM}, volume 3465 of {\em Lecture Notes in Computer Science},
  pages 1--28, 2005.

\bibitem{DBLP:conf/mfcs/NevenSV01}
F.~Neven, T.~Schwentick, and V.~Vianu.
\newblock Towards regular languages over infinite alphabets.
\newblock In {\em MFCS}, volume 2136 of {\em Lecture Notes in Computer
  Science}, pages 560--572, 2001.

\bibitem{P99}
M.~Pistore.
\newblock {\em History Dependent Automata}.
\newblock PhD thesis, University of Pisa, 1999.

\bibitem{pitts-book}
Andrew Pitts.
\newblock {\em Nominal sets: names and symmetry in computer science}, volume~57
  of {\em Cambridge Tracts in Theoretical Computer Science}.
\newblock Cambridge University Press, 2013.

\bibitem{rado}
R.~Rado.
\newblock Universal graphs and universal functions.
\newblock {\em Acta Arith.}, 9:331--340, 1964.

\bibitem{DBLP:conf/csl/Segoufin06}
L.~Segoufin.
\newblock Automata and logics for words and trees over an infinite alphabet.
\newblock In {\em CSL}, volume 4207 of {\em Lecture Notes in Computer Science},
  pages 41--57, 2006.

\bibitem{staton-thesis}
S.~Staton.
\newblock {\em {Name-passing process calculi: operational models and structural
  operational semantics}}.
\newblock PhD thesis, University of Cambridge, 2007.

\end{thebibliography}

\end{document}